\documentclass[11 pt]{article}
\usepackage{geometry}  
\usepackage{amssymb} 
\usepackage{graphicx} 
\usepackage[font=footnotesize, labelfont=bf]{caption} \usepackage{subcaption} \usepackage{booktabs} \usepackage{setspace} \usepackage{xpatch} \usepackage{csquotes} \usepackage{enumitem} \usepackage[section]{placeins} \usepackage[flushleft]{threeparttable} \usepackage{cjhebrew} 
\usepackage{pdflscape} 
\usepackage[section]{placeins} 

\usepackage{titlesec}
\titleformat{\paragraph}
{\normalfont\normalsize\bfseries}{\theparagraph}{1em}{}
\titlespacing*{\paragraph}
{0pt}{3.25ex plus 1ex minus .2ex}{1.5ex plus .2ex}

\usepackage[page,header,title]{appendix}

\usepackage{titletoc}
\usepackage[nottoc,numbib]{tocbibind}
\usepackage{calc}
\usepackage{tocloft}
\renewcommand*{\cftfigname}{\figurename\space}

\renewcommand*{\cfttabname}{\tablename\space}

\usepackage{latexsym}
\usepackage{relsize}
\usepackage{etoolbox}
\usepackage{bm}
\usepackage{color}
\usepackage{soul}
\usepackage{xcolor}
\usepackage{siunitx} \usepackage{amsmath} \usepackage{amsthm} \usepackage{physics} \usepackage{multirow} 
\usepackage{tikz} \tikzstyle{node} = [circle, draw, text centered]
\tikzstyle{line} = [draw, -latex']
\usetikzlibrary{shapes,arrows}
\usetikzlibrary{tikzmark,fit,shapes.geometric} \usetikzlibrary{arrows.meta} \usetikzlibrary{positioning}
\tikzset{>={Latex[width=3mm,length=3mm]}} \newtheorem{result}{Result}
\newtheorem{proposition}{Proposition}

\usepackage{url}
\usepackage{hyperref}
\hypersetup{
    colorlinks=true,
    linkcolor=blue,
    filecolor=magenta,      
    urlcolor=cyan
    }
\usepackage[english]{babel} \usepackage[
    doi=false,     isbn=false,
    url=false,
    eprint=false,
    backend=biber,
    style=authoryear,     maxbibnames=9,
    maxcitenames=2
  ]{biblatex}

\AtEveryBibitem{\clearfield{month}}
\AtEveryBibitem{\clearfield{day}} 
\AtEveryBibitem{\clearlist{language}} 
\renewbibmacro*{volume+number+eid}{  \printfield{volume}  \setunit*{\addnbspace}  \printfield{number}  \setunit{\addcomma\space}  \printfield{eid}}
\DeclareFieldFormat[article]{number}{\mkbibparens{#1}}

\usepackage{csquotes} \renewbibmacro{in:}{  \ifentrytype{article}{}{\printtext{\bibstring{in}\intitlepunct}}}

\begin{document}

\begin{titlepage}
   \begin{center}
        \Large
        \textbf{Horizontal or Vertical? Managerial Attention and Organizational Structure} \\
       \vspace*{3cm}
       Ashley Perry \\
       New York University Abu Dhabi
            
       \vspace{3cm}
        \small
        \begin{abstract}
           \noindent This paper analyzes the relative efficiency of horizontal and vertical organizational structures modeled as communication networks under limited managerial attention. In this framework, workers receive noisy signals about the state of the world and communicate reports to a decision-maker whose meeting time relatively more constrained. The clarity of each report depends endogenously on meeting lengths and the workers' heterogeneous communication and processing abilities. When workers possess identical communication abilities, the horizontal structure is always optimal because it allows the decision-maker to weight signals optimally without compounding transmission noise. Conversely, when workers have asymmetric communication abilities, the attention constraint of the decision-maker induces a structural switch. Under these constraints, the vertical structure becomes preferred because the decision-maker can delegate the time-intensive task of information aggregation to the superior communicator, balancing the loss from transmission noise against the premium on the decision-maker's scarce time. Under the assumption of homogeneous processing abilities I derive the analytical threshold for the \(n\)-worker model. 
        \end{abstract}
\begin{minipage}{\textwidth}
  \raggedright
  \small

\noindent This version: September 2026

\noindent organizational design; communication; information aggregation
         
\end{minipage}
\vfill
\begin{minipage}{\textwidth}
  \raggedright
  \small

  \bigskip
  {\footnotesize \noindent\textit{Acknowledgments.} I'd like to that Jean-Pierre Beno\^it, Andrea Galeotti, Ben Golub, David Myatt, Emre Ozdenoren, as well as three anonymous referees and seminar participants at London Business School, for their valuable comments and suggestions.}
\end{minipage}

   \end{center}
\end{titlepage}

\onehalfspacing

\section{Introduction}
There is a broad view that the structure of organizations can be either hierarchical \parencite{radner1992HierarchyEconomicsManaging,garicano2000HierarchiesOrganizationKnowledge}, or they can be flat \parencite{drucker1988ComingNewOrganization,drucker1992NewSocietyOrganizations, puranam2015ValveWay, lee2017SelfmanagingOrganizationsExploring}. These structures can be viewed as communication networks governing the flow of information in an organization (\cite{bolton1994FirmCommunicationNetwork}). However, what is the best structure for an organization? This paper considers these two structures, horizontal (flat) and vertical (hierarchical), and asks under what conditions would an organization prefer to have one over the other. For example, consider the CEO of a tech start-up firm building a new product. Although she understands the broad market, she is not as informed as her specialized executives, such as lead engineer (technical specifications) or salesperson (market demand). To solve this problem, the CEO must gather information from her team. She could collect this information using two different approaches. Either she could use a horizontal structure, where she meets individually with both the engineer and the salesperson, or she could use a vertical structure, where she meets a single employee, such as a Chief Operating Officer (COO), who has aggregated all the department heads' information before sharing it with her. Because all employees share the common objective of maximizing the firm's success, there are no misaligned incentives. Instead, this paper aims to understand the role that the CEO's scarcity of time and the varying communication abilities of her employees play in determining the optimal structure of the organization.

There is a large literature on organizational structure. Empirically there is evidence that firms have flattened \parencite{rajan2006FlatteningFirmEvidence}, potentially due to increasing product market competition \parencite{guadalupe2010FlatteningFirmProduct}, although organizational structures can vary depending on the function e.g. product versus administrative \parencite{guadalupe2014WhoLivesCSuite}. While much of the theoretical literature focuses on how organizational design mitigates strategic conflict and misaligned incentives \parencite{friebel2004AbuseAuthorityHierarchical,currarini2006DelegationCentralizationRole,migrow2021DesigningCommunicationHierarchies}, this paper isolates a purely operational channel. By analyzing a team-theoretic framework where all agents share a common objective \parencite{marschak1972EconomicTheoryTeams,dessein2016RationalInattentionOrganizational}, I abstract from incentive conflicts to highlight how the attention scarcity and heterogeneous communication skills shape the optimal flow of information. Even if workers have a common objective they may still have private information, for example, they may have different social networks or have expertise in different areas. This setup with a common objective and private information, reflects numerous situations within organizations. For instance, in the motivating example the CEO and her subordinates can have the same goal-- a successful product launch -- but they may be informed differently about how best to achieve that goal.

In order for a decision-maker of an organization to decide on a course of action where she is not fully informed, she can extract information from other potentially better informed employees. The organizational structure of the firm, by dictating who interacts with whom, can determine what and how this information is used. Individuals may communicate their ideas more or less clearly than each other and may not be given equal time to share their ideas, which can affect the decision made. Understanding how these two factors affect the information conveyed is important in deciding how best to organize employees so that optimal decisions are made for the firm. Clearly, how a manager spends her limited time within an organization and with whom she spends it will have an impact on the quality of the decision she makes \parencite[e.g.][]{bandiera2020CEOBehaviorFirm}. However, she may also have other responsibilities, such as raising outside investment in the firm, bidding for new contracts, or meeting stakeholders \parencite{porter2018HowCEOManage}. Hence, she is clearly time constrained and this model tries to highlight the role that communication and scarce attention play in the effectiveness of an organization. Consider the example of the CEO, there will be periods when the CEO is more pressed for time, and depending on who she hires, the ability of her team to articulate their information might not be the same. Both of these factors could affect the quality of the decision made. This paper studies how this scarcity of executive meeting time interacts with employee communication skills to determine whether a decision-maker should communicate directly with departmental leads in a flat structure or rely on an intermediary who aggregates information through a hierarchy.

Formally, I model the organization as a communication network under two distinct structures, horizontal and vertical. Two workers receive noisy signals about the state of the world, where the noise depends on their exogenous processing abilities, which can be thought of as an employee's technical expertise. The workers then communicate these signals to the decision-maker (DM) through meetings, where the clarity of the resulting report is a function of both the meeting's duration and the worker's communication ability, which represents their capacity to summarize and transmit complex information clearly. In the horizontal structure, the decision-maker meets with each worker separately. In the vertical structure, the decision-maker meets only with a designated intermediary, who first meets with the other worker to aggregate their information. Crucially, the decision-maker has limited attention, modeled as a strictly tighter constraint on her total meeting capacity relative to the workers. This framework allows us to analyze how the optimal organizational structure shifts as time becomes increasingly scarce for the head of the firm. 

I characterize these trade-offs in a baseline three-agent setting consisting of one decision-maker and two workers with homogeneous processing abilities. The first result (Result \ref{prop_same_xi}), shows that when workers have the same communication and processing ability, it is always optimal for the firm to have a horizontal organizational structure, regardless of the size of the excess meeting capacity of the other workers. The reason for this is that when communication is the same for each worker, then in the horizontal structure the DM gives optimal attention to the information coming from both workers. However, in the vertical structure, there is extra noise generated by the separate meeting of the workers, which the DM cannot reduce in her only meeting, and so she is unable to optimally weight the signals. Hence, the horizontal structure will always perform better by minimizing the noise to signal ratio. 
The second result (Result \ref{prop_diff_xi}), shows that when the workers have the same processing ability but differing communication abilities and the best communicator meets with the DM directly, then there is a large enough excess meeting capacity of workers for which it is optimal for the firm to switch from the horizontal to the vertical organizational structure. Now with differing communication abilities in the horizontal structure, to offset the noise from this differing ability, the DM pays more attention to the better communicator. This means the DM is not equally weighting the underlying signals and, as the processing ability of the workers has not changed, this is sub-optimal. However, this sub-optimal weighting can be overcome in the vertical structure if the workers meet for a sufficient amount of time to re-balance the weighting of the signals. In other words, there is a point at which the information transmission from the workers meeting is high enough that the preference flips from the horizontal to the vertical structure. Applying these results to the CEO and COO scenario, it is clear that if the CEO has a quiet period it makes sense to speak directly to all her team leads in order to make a decision. But, if she is busy and so has limited time, it makes sense for the COO, to speak to all the team leads, aggregate the information, and then communicate this to the CEO.

I subsequently show in Result (\ref{prop:asymmetric_rhos}) that the assumption of homogeneous processing abilities is not necessary in the baseline three-agent model. Finally, in Result (\ref{prop_general}) I show that the main insight of Results (\ref{prop_same_xi}) and (\ref{prop_diff_xi}), with homogeneous processing abilities, can be extended to the to \(n\)-agents.

\section{Relation to Literature}
\label{litrev}

This paper bridges several distinct strands of the organizational economics literature by embedding Bayesian signal extraction into a team-theoretic model of firm structure.

First, this paper contributes to the literature on information processing within organizations. Organizations have been modeled as hierarchies to minimize the time employees spend processing information, leading to specialized communication networks \parencite{radner1992HierarchyEconomicsManaging,bolton1994FirmCommunicationNetwork} or to efficiently acquire and route the knowledge necessary to solve production problems \parencite{garicano2000HierarchiesOrganizationKnowledge}. These studies emphasize hierarchical structures as an outcome to  efficient solve problems, they largely treat organizational links as frictionless. In contrast, this paper which models transmission noise as an endogenous function of meeting duration and worker-specific communication ability, isolating how managerial attention bottlenecks dictate the optimal structure of internal communication. More recent work has studied managerial attention constraints within organizations \parencite{dessein2016RationalInattentionOrganizational}. Although the focus is on how rational inattention limits organizational coordination, rather than the optimal organizational structure given attention constraints.

Second, by adopting a team-theoretic approach in the tradition of \textcite{marschak1972EconomicTheoryTeams}, this paper separates operational communication frictions from strategic misalignment. A large body of work focuses on how specific structures help overcome misaligned incentives and agency problems between workers and decision-makers \parencite{friebel2004AbuseAuthorityHierarchical,currarini2006DelegationCentralizationRole,migrow2021DesigningCommunicationHierarchies}. While strategic distortion is a central concern in organizational design, focusing on incentive conflicts can obscure purely operational trade-offs. Setting aside strategic misrepresentation demonstrates that communication heterogeneity and managerial time scarcity alone are sufficient to generate structural switches between flat and hierarchical communication channels, even among fully aligned agents. \textcite{geanakoplos1991TheoryHierarchiesBased} also adopt a team-theoretic approach to studying organizational hierarchies, they focus on optimal hierarchies for costly knowledge acquisition with potentially heterogeneous skilled managers, rather than the role of communication frictions and limited managerial attention. 

Finally, the underlying information structure in this model builds upon the Bayesian approach to signal extraction with quadratic loss functions \parencite[][]{kay2013FundamentalsStatisticalSignal,morris2002SocialValuePublic}. Within this setting, prior studies investigate how agents invest attention to coordinate actions or capture peer influence in social and political networks \parencite{calvo-armengol2015CommunicationInfluenceCommunication,dewan2007LeadingPartyCoordination}. This paper adapts that machinery to organizational networks, endogenizing the precision weights chosen by a constrained decision-maker and demonstrating how intermediate aggregators optimally synthesize multilayered signals.

Finally, this theoretical framework provides an informational micro-foundation for empirical patterns in executive time use and organizational flattening. Empirical studies document widespread flattening of corporate hierarchies in response to competitive shocks \parencite{rajan2006FlatteningFirmEvidence,guadalupe2010FlatteningFirmProduct}, as well as significant functional variation across administrative and product divisions \parencite{guadalupe2014WhoLivesCSuite}. Also, data on executive time allocation reveal that top managers face acute time scarcity, balancing internal operational meetings against critical external demands \parencite{mintzberg1973NatureManagerialWork,kotter1990WhatLeadersReally,porter2018HowCEOManage,bandiera2020CEOBehaviorFirm}. Management literature similarly debates whether flat structures enhance responsiveness or create operational paralysis in scaling firms \parencite{drucker1988ComingNewOrganization,drucker1988ComingNewOrganization, puranam2015ValveWay, burton2017GitHubExploringSpace, lee2017SelfmanagingOrganizationsExploring,lee2022MythFlatStart}. This paper rationalizes these observed patterns, flat structures are optimal when executives enjoy ample internal meeting time, but hierarchical aggregation becomes indispensable when managerial attention is scarce and technical specialists require articulation through skilled intermediaries.

\section{Model}
\label{model}
\textbf{Agents and Structure.} There are two types of agents; A DM who is the head of the organization and workers who are subordinates to the DM. I assume, for now, that there are only two workers. This small number is sufficient to demonstrate my main findings, later I will show that the model can be extended to a more general setting with $n$ workers. 

I consider two organizational structures; First, a horizontal structure, where the DM is directly connected to the workers (see Figure \ref{fig:central}). Then a vertical structure, where the DM is directly connected to a single worker, who then is also connected to the other worker (see Figure \ref{fig:decentral}). The arrowhead indicates the direction in which the information is flowing. 

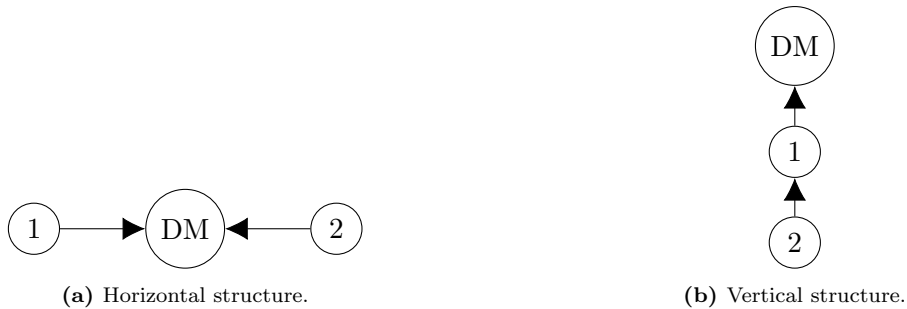
\begin{figure}[h]
\centering
    \begin{subfigure}[b]{0.45\textwidth}
        \centering
        \begin{tikzpicture}[node distance = 1cm, auto]
            \node [node] at (0,0) (DMa) {DM};
            \node [node] at (-2,0) (onea) {1};
            \node [node] at (2,0) (twoa) {2};
            \draw[->] (twoa) -- (DMa);
            \draw[->] (onea) -- (DMa);
        \end{tikzpicture}
        \caption{Horizontal structure.} \label{fig:central}
    \end{subfigure} \qquad 
    \begin{subfigure}[b]{0.45\textwidth}
        \centering
        \begin{tikzpicture}[node distance = 1cm, auto]
            \node [node] at (0,0) (DM) {DM};
            \node [node] at (0,-1.4) (one) {1};
            \node [node] at (0,-2.6) (two) {2};
            \draw[->] (two) -- (one);
            \draw[->] (one) -- (DM);
        \end{tikzpicture}
        \caption{Vertical structure.} \label{fig:decentral}
    \end{subfigure}
    
    \caption{Organizational Structures} \label{fig:main_structures}
\end{figure}

\noindent \textbf{Information.} Each agent is endowed with a fixed amount of meeting time $\kappa$. All agents in the organization share a common prior belief about the unknown state of the world $\theta$. This prior is normally distributed as follows: $\theta \thicksim   \mathcal{N}(\mu_\theta, \sigma^2_\theta)$. Workers receive an individual signal about the state of the world with some normally distributed noise $\epsilon_i$ for $i=\{1,2\}$. The parameter $\rho_i > 0$ captures the processing ability of the worker, the higher $\rho_i$ the lower the variance of the error term $\epsilon_i$ and the more accurate the signal. Hence, the more able the worker is in processing the signal they receive, the clearer the signal is. The $\rho_i's$ are taken as exogenous. The information contained within each signal for the worker $i$ is given by \(
s_i = \theta + \epsilon_i, \text{where}\  \epsilon_i \thicksim \mathcal{N} \left( 0 , \frac{1}{\rho_i}\right)\).

All agents in the organization share a common prior belief about the unknown state of the world $\theta \sim \mathcal{N}(\mu_\theta, \sigma_\theta^2)$. Each subordinate worker $i \in {1, 2}$ receives a private noisy signal $s_i = \theta + \epsilon_i$, where the error term $\epsilon_i \sim \mathcal{N}(0, 1/\rho_i)$ is independent across workers and the parameter $\rho_i > 0$ represents worker $i$'s exogenous technical processing ability. The activity of estimating the true $\theta$ (Section \ref{model}) can be thought of as the primary activity, such as production decisions and internal promotions.

Every agent in the organization is endowed with a total working time capacity $\kappa > 0$. However, agents differ in their internal availability. While subordinate workers dedicate their entire endowment $\kappa$ to internal operational communication, the DM must divide her time between internal oversight and essential external demands, such as securing outside capital, negotiating supplier contracts, or managing external stakeholders. Consequently, the DM faces an effective internal meeting capacity constraint $\bar{m}$, where $0 < \bar{m} < \kappa$. Managerial attention is therefore strictly scarcer than subordinate operational time, establishing a structural bottleneck at the apex of the organization. This is consistent with evidence that the higher up the organization an employee is, the more likely they are to have a greater number of responsibilities than the organizations junior colleagues \parencite{mintzberg1973NatureManagerialWork,kotter1990WhatLeadersReally}. 

This is not meant to be a comprehensive characterization, as clearly there could be some interdependence between the primary activity and other external activities. It is only meant to illustrate that a DM has a fixed amount of time to spend on all the activities her job may entail, and how she allocates this time has consequences for the firm.

\section{Horizontal Structure}

\noindent \textbf{Meeting Time.} Given the horizontal structure (see Figure \ref{fig:central}), workers can interact with the DM, but cannot do so with each other. As the workers are assumed to have the same objective as the DM, the workers do not try to bias the DM's decision. However, their signal is communicated to the DM along with some additional normally distributed noise $\eta_i$. I will refer to this as the report of the meeting. The DM can allocate varying amounts of meeting time $m_i \geq 0$, to determine the clarity of the report by affecting the variance of the noise term $\eta_i$. For example, the interaction between the DM and worker 1 can be thought of as a meeting and $m_1$ captures the duration of that meeting. The maximum meeting time that can be chosen is constrained by the DM's meeting time to $\bar{m}$.

In addition to meeting time, the clarity of the report is affected by a communication friction parameter, $\xi_i^2$. Economically, this parameter captures the ``social distance'' or lack of shared context between the worker and the DM. For instance, a high $\xi_i^2$ could reflect a highly specialized technical worker who struggles to translate domain-specific jargon into terms the DM can easily process, whereas a low $\xi_i^2$ represents a worker with high organizational literacy or a shared professional background with the DM. This communication friction plays a pivotal role in determining the optimal organizational structure. The reports are of the following form \(r_i = s_i + \eta_i,\ \text{where}\ \eta_i \thicksim \mathcal{N} \left( 0 , \frac{\xi^2_i}{m_i} \right)\). The random variables $\{\theta, \epsilon_i,\eta_i\}$ are all independent of each other. For instance, for workers $1$ and $2$ the co-variance of the error terms in their signals would be zero.  \\

\noindent \textbf{Action and Payoffs.} 
The DM maximizes expected organizational payoff by minimizing expected quadratic loss over a two-stage decision process. Ex-ante, prior to the realization of states and signals, the decision-maker chooses the allocation of meeting durations $(m_1, m_2)$ to minimize the expected variance of the posterior estimate subject to her aggregate meeting capacity. Ex-post, upon observing the realized report vector $\mathbf{r} = (r_1, r_2)$, she selects an action $a \in \mathbb{R}$ conditional on her information set, yielding the standard Bayesian estimate $a^*(\mathbf{r}) = \mathbb{E}[\theta \mid \mathbf{r}]$. Ex-ante, prior to the realization of the state and transmission shocks, the decision-maker chooses meeting durations $(m_1, m_2)$ to minimize the integrated quadratic risk subject to the managerial attention budget. Formally, the ex-ante optimization problem is given by 
\begin{equation}
\label{DM_optimization}
	\begin{split}
		\max_{m_1,m_2\geq0} &\bar{u}-\mathbb{E} \big[ (\theta -a^*(\mathbf{r}))^2 | \big] \\
		 & \text{s.t.} \hspace{3mm} m_1+m_2 \leq \bar{m} .
	\end{split}
\end{equation}
To ensure well-defined precision measures across all feasible allocations, including potential corner solutions where $m_i = 0$, the effective precision of worker $i$'s communication channel is defined continuously on $\mathbb{R}+$ by $h_i(m_i) = \frac{\rho_i m_i}{\xi_i^2 \rho_i + m_i}$ for $m_i > 0$, with continuous extension $h_i(0) \equiv \lim{m_i \to 0^+} h_i(m_i) = 0$. When $m_i = 0$, the DM conducts no meeting with worker $i$, the transmission noise variance tends to infinity, and worker $i$'s report carries zero precision weight in the posterior expectation.\\

\noindent \textbf{Timing.} The DM decides on the length of the meetings $\textbf{m}$. The workers receive their signal, whose accuracy is determined by the given $\rho_i$. The workers then meet the DM and a report is produced, for which the clarity depends on the chosen $m_i$ and given $\xi_i^2$. The DM observes the vector of reports $\mathbf{r}$ and then decides on an action $a$ which maximizes her utility $U_H$. \\

\noindent The expected payoff of the DM is $U_H = \bar{u} - \mathbb{E} \big[ (\theta -a)^2 | \mathbf{r}\big] $ and assuming an arbitrary meeting allocation $\textbf{m}$ and reports $\mathbf{r}$, the first-order condition with respect to $a$ yields $a^*=\mathbb{E} [\theta | \mathbf{r}]$. Given the linear structure of the messages, the normally distributed error terms, and that the posterior PDF $p\left(\theta|\mathbf{r}\right)$ is Gaussian, the optimal action $a^*$ will be a linear action. The linear structure of $a^*$ and the independence of the error terms simplifies her objective function, as given in the following proposition. Also, assume that the DM is uninformed i.e. has a diffuse prior $h_0 \to 0$ (or equivalently, $\sigma_\theta^2 \to \infty$). In this limit, the prior precision vanishes, yielding the characterized linear action rule and expected utility

\begin{proposition}
\label{EquilActCentral}
Let there be an arbitrary allocation of meeting time $\{m_1, m_2\}$ and a diffuse prior $\sigma_\theta^2\rightarrow\infty$. Then, given the reports $\{r_1,r_2\}$, the optimal action and the expected payoff of the DM is given by 
\begin{equation*}
	\begin{split}
		a^*=\sum_1^2\omega_i r_i \ \text{\&} \ U_H=\bar{u}-\frac{1}{\sum_1^2 h_i},\ \text{where}\ \omega_i = \frac{\frac{ \rho_i m_i}{\xi^2_i\rho_i + m_i}}{\sum_{i=1}^2 \frac{ \rho_i m_i}{\xi^2_i\rho_i + m_i}}\ \text{\&}\ h_i=\frac{\rho_i m_i}{\xi^2_i\rho_i + m_i} \quad \forall i=\{1,2\}
	\end{split}
\end{equation*}	
\end{proposition}
\noindent The proof for this and all subsequent results are found in the Appendix \ref{appA}. This result shows that the optimal action for the DM is a weighted linear combination of the reports she receives. Here, the weight $\omega_i$ is the relative precision that the DM assigns to a given report. The utility of DM increases with the precision weight $h_i$ of each worker's report.

\section{Vertical Structure}

\noindent \textbf{Information and Meeting Time.} In the vertical structure, information aggregation occurs sequentially. In the vertical structure the DM can no longer meet with worker 2, but only with worker 1 (see Figure \ref{fig:decentral}). Thus, she now decides on how long she should meet with worker 1, $m_1$, and how long worker 1 should meet with worker 2, $m_2$. Now, the DM and worker 1's meeting is constrained to \(m_1\leq \bar{m}\), whereas \(m_2\) is only subject to the constraint that $\sum_1^2 m_i \leq \kappa$. One can interpret \(\kappa\) as the working day. Information is passed through the organization as follows, worker 1 processes the information available to him, after his meeting with worker 2, and then transmits this to the DM in their meeting. The other details, such as the independence of the error terms, are kept the same.  \\

\noindent \textbf{Timing, Action and Payoffs.} The DM decides on the length of the meetings $\textbf{m}$. The workers receive their signal, whose accuracy is determined by the given $\rho_i$. The worker 1 and 2 meet and produce a report. Then worker 1 meets the DM and produces a report based on the information available. These reports depend on the chosen $m_i$ and given $\xi_i^2$. The DM observes the final report, worker 1's potentially incorporating worker 2's, and then decides on an action $a$ which maximizes her utility. The report of the meeting of worker 2 with worker 1, $r_{2,1}$, will be the signal of worker 2 and some normally distributed noise $\epsilon_2$ depending on their processing ability $\rho_2$, plus some normally distributed noise $\eta_2$, dependent on the amount of time they were allocated $m_2$ and their communication ability $\xi^2_2$. Due to the independence of the error terms, this report has the following structure \(r_{2,1}= \theta+ \mathcal{N} \big( 0 ,  \frac{1}{\rho_{2}}+\frac{ \xi^2_2}{ m_{2}}\big).\)  

Worker 1's information set will be their own signal $s_1$ and report $r_{2,1}$. Worker 1's message is required to be an unbiased linear estimate of
$\theta$. Thus, write
\begin{equation*}
    M_1(s_1,r_{2,1})=q s_1+(1-q)r_{2,1},
\end{equation*}
so the signal coefficients sum to one, $q\in[0,1]$. The intercept is set to zero without loss of generality. This normalization fixes the scale of the message relative to the additive transmission noise. Worker 1 then transmits a message $M_1(s_1, r_{2,1})$ to the decision-maker through a meeting of duration $m_1$, which adds transmission noise $\eta_1 \sim \mathcal{N}(0, \xi_1^2 / m_1)$, generating the final report $r_{1,dm} = M_1(s_1, r_{2,1}) + \eta_1$. Because all organizational members share the common objective of maximizing expected utility \parencite{marschak1972EconomicTheoryTeams}, worker 1 chooses \((q,1-q)\) to minimize the DM's expected quadratic loss. The following proposition jointly establishes worker 1's optimal reporting strategy and the DM's optimal action rule. 

\begin{proposition}
\label{EquilActDecent}
Let there be an arbitrary allocation of meeting time $\{m_1, m_2\}$ and diffuse prior $\sigma_\theta^2 \rightarrow\infty$. Worker 1's optimal aggregation strategy is the conditional expectation of the state given his available information \(M_1(s_1, r_{2,1}) = \mathbb{E}[\theta \mid s_1, r_{2,1}] = q_{1,1} s_1 + q_{1,2} r_{2,1},\) where the optimal weights are proportional to signal precisions 
\begin{equation*}
q_{1,1} = \frac{\rho_1}{\rho_1 + \frac{m_2 \rho_2}{m_2 + \xi_2^2 \rho_2}}, \quad q_{1,2} = \frac{\frac{m_2 \rho_2}{m_2 + \xi_2^2 \rho_2}}{\rho_1 + \frac{m_2 \rho_2}{m_2 + \xi_2^2 \rho_2}}
\end{equation*}
The DM's optimal action is $a^*(r_{1,dm}) = r_{1,dm}$, and the resulting expected payoff is 

\begin{equation*}
    U_V = \bar{u} - \left( \frac{1}{\rho_1 + \frac{m_2 \rho_2}{m_2 + \xi_2^2 \rho_2}} + \frac{\xi_1^2}{m_1} \right)
\end{equation*}
 	
\end{proposition}

Because all members share the common objective of maximizing organizational payoff, worker 1 selects an unbiased linear transmission rule $M_1(s_1, r_{2,1}) = q_{1,1}s_1 + q_{1,2}r_{2,1}$ with $q_{1,1} + q_{1,2} = 1$, and the decision-maker selects an action $a^*(r_{1,dm})$ to minimize the expected quadratic loss $\mathbb{E}[(\theta - a)^2]$. Worker 1 constructs precision measures for each piece of information, treating his private signal as the first observation and the report from worker 2 as the second observation. The precision measures are $h_{1,1} = \rho_1$ and $h_{1,2} = \frac{m_2\rho_2}{m_2 + \xi_2^2\rho_2}$, where $h_{1,2}$ is the inverse variance of $r_{2,1}$. Under a diffuse common prior ($\sigma_\theta^2 \to \infty$), the prior precision is zero, so worker 1’s posterior belief relies entirely on the observed signals with relative weights $q_{1,1} = \frac{h_{1,1}}{h_{1,1} + h_{1,2}}$ and $q_{1,2} = \frac{h_{1,2}}{h_{1,1} + h_{1,2}}$. Worker 1 then transmits this posterior estimate to the decision-maker through a meeting of duration $m_1$, generating the final report $r_{1,dm} = \mathbb{E}[\theta \mid s_1, r_{2,1}] + \eta_1$, which yields the vertical expected payoff $U_V$. See Figure \ref{fig:timing_noise} for a summary of the timing and communication in each organization structure in this setup.

\begin{figure}[htbp]
\centering
\begin{subfigure}[t]{0.48\textwidth}
\centering
\begin{tikzpicture}[>=Latex, node distance=0.65cm]
\node[draw, rounded corners, align=center, text width=3.5cm] (h0) {DM chooses $m_1,m_2$\\$m_1+m_2\leq\bar m$};
\node[draw, rounded corners, align=center, text width=3.5cm, below=of h0] (h1) {Workers observe\\$s_i=\theta+\epsilon_i$};
\node[draw, rounded corners, align=center, text width=3.5cm, below=of h1] (h2) {Direct reports\\$r_i=s_i+\eta_i$\\$\mathrm{Var}(\eta_i)=\xi_i^2/m_i$};
\node[draw, rounded corners, align=center, text width=3.5cm, below=of h2] (h3) {DM observes $(r_1,r_2)$\\and chooses $a$};
\draw[->] (h0)--(h1);
\draw[->] (h1)--(h2);
\draw[->] (h2)--(h3);
\end{tikzpicture}
\caption{Horizontal timing.}
\end{subfigure}
\hfill
\begin{subfigure}[t]{0.48\textwidth}
\centering
\begin{tikzpicture}[>=Latex, node distance=0.55cm]
\node[draw, rounded corners, align=center, text width=3.7cm] (v0) {DM chooses $m_1,m_2$\\$m_1\leq\bar m$, $m_1+m_2\leq\kappa$};
\node[draw, rounded corners, align=center, text width=3.7cm, below=of v0] (v1) {Workers observe\\$s_i=\theta+\epsilon_i$};
\node[draw, rounded corners, align=center, text width=3.7cm, below=of v1] (v2) {Worker 2 reports to worker 1\\$r_{2,1}=s_2+\eta_2$};
\node[draw, rounded corners, align=center, text width=3.7cm, below=of v2] (v3) {Worker 1 forms\\$\mathbb{E}[\theta\mid s_1,r_{2,1}]$};
\node[draw, rounded corners, align=center, text width=3.7cm, below=of v3] (v4) {Worker 1 reports to DM\\$r_{1,dm}=\mathbb{E}[\theta\mid s_1,r_{2,1}]+\eta_1$};
\node[draw, rounded corners, align=center, text width=3.7cm, below=of v4] (v5) {DM observes $r_{1,dm}$\\and chooses $a$};
\draw[->] (v0)--(v1);
\draw[->] (v1)--(v2);
\draw[->] (v2)--(v3);
\draw[->] (v3)--(v4);
\draw[->] (v4)--(v5);
\end{tikzpicture}
\caption{Vertical timing.}
\end{subfigure}
\caption{\textit{Notes}: Timing and communication channels. Each transmission by worker $i$ adds variance $\xi_i^2/m_i$.}
\label{fig:timing_noise}
\end{figure}
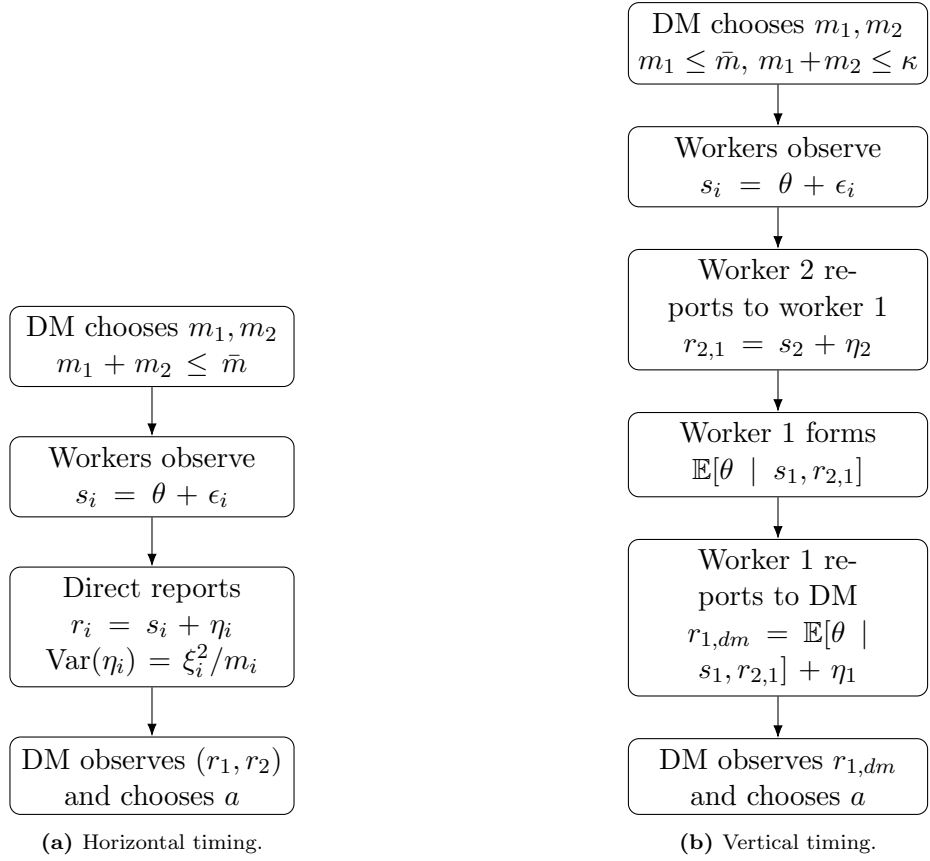

It should be noted that in this model, I have not imposed any restrictions on the minimum time needed to meet to understand a report, which would have implications for the preferred structure. For example, if it was the case that a person needed to be in a meeting for 30 minutes to understand a report, then \(\bar{m}=30\) minutes. In the horizontal structure each of the DM's meetings would be 15 minutes, whereas in the vertical structure her only meeting would be 30 minutes. Given the lower bound on learning, it is clear that it would be better to have one meeting of 30 minutes than two of 15 minutes. However, I do not include this because I want to highlight the role that the interaction of scarcity of attention and communication ability plays in determining which organizational structure is better for information transmission.

\section{Results}
\label{results}
\begin{result}
\label{prop_same_xi}
Suppose $\rho_1=\rho_2=\rho$. If communication ability is symmetric, $\xi_1^2=\xi_2^2=\xi^2$, then for every $0<\bar m<\kappa<\infty$, the horizontal structure strictly dominates the vertical structure. As $\kappa\to\infty$ the communication noise in the structures converge.

\end{result}

\noindent The implication of Result \ref{prop_same_xi} is that in the vertical setting, even if worker 1 and worker 2 meet for a \textit{very} long time to minimize the communication noise between them, the performance of this structure is still worse than that of the horizontal one. This is surprising since in the vertical structure the total time spent in meetings can be arbitrarily greater than in the horizontal structure, which is restricted to \(\bar{m}\), so it might seem that communication noise would be less in the vertical structure, as more time can be spent by worker 1 and worker 2 in their meeting to produce as clear a report as possible. Figure \ref{fig:same_ability} shows the difference in payoff between the two structures when communication ability (\(\xi\)) and processing ability (\(\rho\)) are the same. It illustrates that as $\kappa$ increases and the worker excess meeting time increases, this decreases the payoff difference between the two structures, which converge in the limit.

\begin{figure}[htbp]
	\centering
	\begin{subfigure}[t]{0.49\textwidth}
		\centering
		\includegraphics[width=\textwidth]{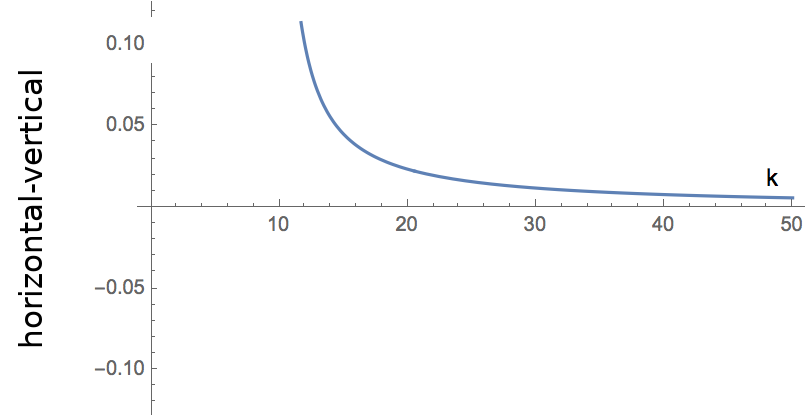}
		\caption{Parameters: $\{\bar{m}, \rho, \xi^2_1, \xi^2_2 \}=\{10,1,1,1 \}$.}
		\label{fig:same_ability}
	\end{subfigure}
	\hfill 	\begin{subfigure}[t]{0.49\textwidth}
		\centering
		\includegraphics[width=\textwidth]{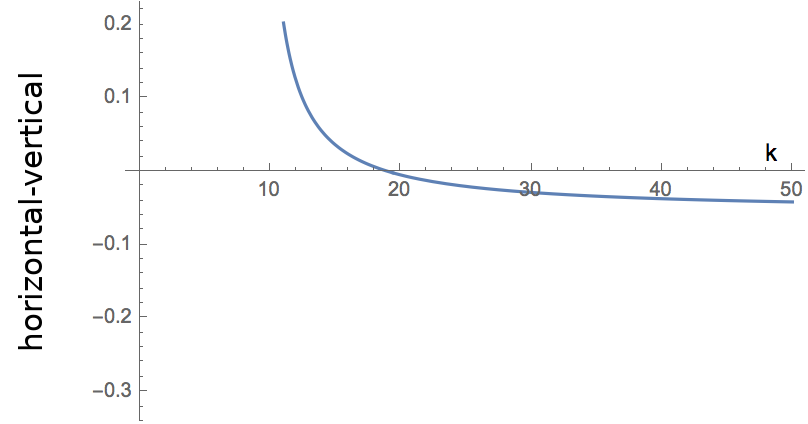}
		\caption{Parameters: $\{\bar{m}, \rho, \xi^2_1, \xi^2_2 \}=\{10,1,1,1.5 \}$.}
		\label{fig:diff_ability}
	\end{subfigure}
	
	\caption{\textit{Notes}: These figures plot the difference in expected payoff between the horizontal and vertical structures ($U_H - U_V$) on the y-axis against the total available meeting time ($\kappa$) on the x-axis. In (a) as workers have identical communication abilities, $\xi^2_1= \xi^2_2$, the payoff difference remains strictly positive. The horizontal structure dominates the vertical structure for all values of $\kappa$, as the DM can optimally weight the signals without relying on an intermediate aggregator. In (b) as worker 1 is a better communicator, $\xi^2_1< \xi^2_2$, there is a threshold point $\bar{\kappa}$ above which the payoff difference becomes negative. For $\kappa > \bar{\kappa}$, the vertical structure outperforms the horizontal as worker 1 has sufficient time to extract signal form the noisy communication of worker 2.}
	\label{fig:overall_organization_plots} \end{figure}

However, if I alter the communication ability so that worker 1 is relatively better at communicating than worker 2. Then, there exists point at which the DM's capacity constraint means the vertical structure is preferred to the horizontal one, as the following proposition states.
\begin{result}
\label{prop_diff_xi}
Suppose $\rho_1=\rho_2=\rho$ and $0<\xi_1<\xi_2$. Let $d\equiv \xi_2-\xi_1$. The vertical structure satisfies $m_1\leq \bar m$ and $m_1+m_2\leq\kappa$. Its optimal allocation is
\begin{equation}
m_1^*(\kappa)=\min\left\{\bar m,\, \frac{\xi_1(2\kappa+\rho\xi_2^2)}{2\xi_1+\xi_2}\right\}, \qquad m_2^*(\kappa)=\kappa-m_1^*(\kappa).
\end{equation}

\begin{enumerate}
    \item If $\bar m\leq \rho\xi_1d$, the horizontal allocation is a corner solution: $m_1^H=\bar m$ and $m_2^H=0$. The two structures have equal payoffs at $\kappa=\bar m$, and the vertical structure strictly dominates for every $\kappa>\bar m$. In this case define $\bar\kappa=\bar m$.

    \item If $\bar m>\rho\xi_1d$, the horizontal allocation is interior. Set
\begin{equation}
\bar\kappa=
\begin{cases}
\kappa_I, & \kappa_I\leq\kappa_B,\\
\kappa_C, & \kappa_I>\kappa_B.
\end{cases}
\end{equation}
Then $\bar\kappa>\bar m$. The two structures have equal payoffs at $\kappa=\bar\kappa$, and the vertical structure strictly dominates the horizontal structure if and only if $\kappa>\bar\kappa$.
\end{enumerate}

\end{result}

\noindent These two results reflect the interplay between managerial attention and communication. To understand the intuition for the results, first consider Result \ref{prop_same_xi}. The environment be symmetric in all respects ($\rho_1=\rho_2$ and $\xi_1^2=\xi_2^2$). Hence, given their identical communication abilities of the workers, in the horizontal structure the DM gives equal attention to them both, e.g. the meeting time is \(\bar{m}/2\) for both workers. Note that the available meeting time is exhausted as the DM's utility is increasing in meeting time. This means that she is equally weighting the underlying signals of the reports, which given their identical processing ability is optimal. However, in the vertical structure the DM is sub-optimally weighting the signals as she is giving less weight to worker 2's signal. This is true even as the workers are given relatively more time to meet, e.g. as \(\kappa\) increases. This is because the noise generated from the report between the workers is inversely proportional to the length of their meeting and for a finite \(\kappa\) it will not be eliminated. Hence, the noise-to-signal ratio is always higher in the vertical organization. Moreover, given the constraint on the organizational meeting time, the more time is allocated to workers' meeting to improve the clarity of that report, the less time is allocated to the DM's only meeting, which reduces the clarity of the aggregated report given to the DM.

For Result \ref{prop_diff_xi} worker 1 faces lower communication frictions than worker 2 ($\xi_1^2 < \xi_2^2$). In a real-world organization, worker 1 might represent a generalist manager who shares a common vocabulary with the DM, while worker 2 represents a technical specialist -- such as a lead engineer or salesperson -- whose raw insights are highly accurate but difficult for the DM to parse directly. The DM's direct meetings with all workers forces her to allocate more of her constrained time \(\bar{m}\) towards worker 1 to reduce the distortion or noise in the information communicated to her. When the managerial attention is severely scarce (\(\bar{m}\leq \rho \xi_1(\xi_2-\xi_1)\)), the distortion is large enough that it is optimal to completely ignore worker 2 and only meet worker 1. The vertical structure overcomes this by delegating the time-intensive task to the meeting between worker 1 and 2, leveraging their excess meeting capacity \(\kappa-\bar{m}\) to incorporate worker 2's information into worker 1's report to the DM. When, on the other hand, the DM is not as severely constrained \(\bar{m}> \rho \xi_1(\xi_2-\xi_1)\), both workers now meet the DM in the horizontal structure, but the DM will pay more attention to worker 1 than to worker 2 to offset this communication friction. The DM will meet both workers for a total of $\bar{m}$ time, but that will not be equally split, which means that the DM is not equally weighting the underlying signals. This is sub-optimal since the processing ability has not changed and the total noise generated in the horizontal structure is the same as under the symmetric case. 

However, this sub-optimal weighting can be overcome in the vertical structure if worker 1 and 2 can meet for a sufficient amount of time to re-balance the weighting of the signals, then the weighting of signals can be optimized and the DM can extract more information. That is, if there is sufficient excess meeting capacity \(\kappa-\bar{m}\) beyond some critical threshold \(\bar{\kappa}\), the extra clarity obtained from the meeting of worker 1 and 2 outweighs the transmission noise, making the vertical structure strictly optimal. Demonstrating that when time is scarce for the head of an organization, they would do well to have her best communicator directly speak to them once they have met the other workers and aggregated all the relevant information, rather than meeting everyone individually to gather this information.

Figure \ref{fig:diff_ability} shows the difference in payoff between the two structures when worker 1 has a better communication ability than worker 2, along with equal processing ability (\(\rho\)) for illustrative purposes. It also shows that as $\kappa$ increases, the difference between the two structures decreases, but now there is a point $\bar{\kappa}$ where the vertical is preferred to the horizontal for \(\kappa > \bar{\kappa}\).

The preference for these two structures is related to discussions on organizational structure in the modern workplace (\cite{drucker1988ComingNewOrganization}, \cite{drucker1992NewSocietyOrganizations}, \cite{puranam2015ValveWay} \cite{lee2017SelfmanagingOrganizationsExploring}). Consider start-ups, there is a view that they are best to initially start of in a horizontal structure to increase creative idea generation through more diverse sources of information and less bureaucracy and then switch to more vertical as they scale (\cite{burton2017GitHubExploringSpace}). However, this might not always be optimal, as horizontal structures can hamper a firm's ability to actually execute their ideas, a necessary condition for them to then scale (\cite{lee2022MythFlatStart}).

Result \ref{prop_same_xi} applies to knife edge case, \(\xi_1^2 =\xi_2^2\) , but the result can be shown to apply more generally. That is, if the necessary assumptions are still met, but $\xi_1^2\geq\xi_2^2$ then the horizontal structure is preferred to the vertical structure. Figure \ref{fig:K_region} shows that if $\xi_1^2 <\xi_2^2$ as $\xi_1^2\to\xi_2^2$ the threshold $\bar{\kappa}$ increases. This is seen in the shaded blue region that corresponds to a threshold of $\bar{\kappa}=10$, whereas the smaller orange region corresponding to a threshold of $\bar{\kappa}=20$. The white region is the parameter space where there is no threshold and this is the region to the right of the \ang{45} line, e.g. $\xi_1^2\geq\xi_2^2$.

\begin{figure}[h]
    \centering
    \begin{minipage}[c]{0.4\textwidth}
        \centering
        \includegraphics[width=\textwidth]{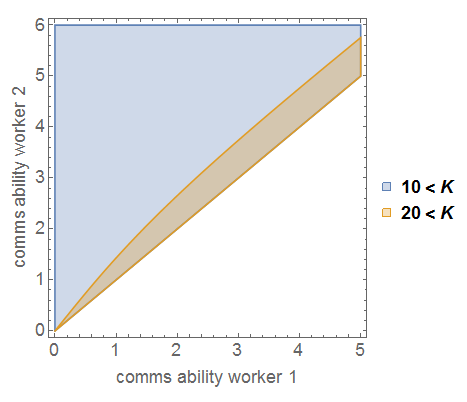}
    \end{minipage}
    \hfill     \begin{minipage}[c]{0.54\textwidth}
        \caption{\textit{Notes}: This figure depicts the organizational choice across worker communication friction parameters, with worker 1's friction $\xi_1^2$ on the horizontal axis and worker 2's friction $\xi_2^2$ on the vertical axis. Parameters are fixed at managerial attention capacity $\bar{m} = 10$ and symmetric processing precision $\rho_1 = \rho_2 = 1$. The 45-degree line ($\xi_1^2 = \xi_2^2$) separates the symmetric benchmark from asymmetric communication. In the unshaded region to the right of the 45-degree line ($\xi_1^2 \ge \xi_2^2$), the horizontal network strictly dominates for all finite meeting capacities $\kappa \in (\bar{m}, \infty)$. In the shaded regions to the left of the 45-degree line ($\xi_1^2 < \xi_2^2$), worker 1 possesses superior communication ability, establishing a finite threshold $\bar{\kappa}$ above which the vertical hierarchy becomes optimal. The dark shaded region represents the structural switch occurs under moderate meeting endowments ($\bar{\kappa} = 10$), whereas the lighter shaded region requires substantial organizational meeting capacity ($\bar{\kappa} = 20$) for hierarchical aggregation to surpass direct horizontal reporting.}

        \label{fig:K_region}
    \end{minipage}
\end{figure}

\noindent As the next result shows the assumption of equal processing abilities is not a necessary requirement for the switch in the preference of communication structures.

\begin{result}
\label{prop:asymmetric_rhos}
Suppose worker processing abilities $\rho_1, \rho_2 > 0$ are heterogeneous. If communication abilities are symmetric ($\xi_1^2 = \xi_2^2 = \xi^2$), then for every finite organizational meeting capacity $\kappa \in (\bar{m}, \infty)$, the horizontal structure strictly dominates the vertical structure.

If worker 1 is a superior communicator ($\xi_1 < \xi_2$) and managerial meeting time is severely constrained ($\bar{m} \le \rho_1 \xi_1(\xi_2 - \xi_1)$), direct communication with worker 2 collapses ($m_2^H = 0$), and the vertical structure strictly dominates the horizontal structure for all $\kappa > \bar{m}$. If managerial meeting time is less constrained ($\bar{m} > \rho_1 \xi_1(\xi_2 - \xi_1)$), both direct meetings occur in the horizontal structure, and there exists a unique finite threshold $\bar{\kappa} > \bar{m}$ such that the vertical structure strictly dominates the horizontal structure if and only if $\kappa > \bar{\kappa}$.
\end{result}

\noindent The intuition for the result is the following. With equal communication ability \(\xi_1^2=\xi_2^2=\xi^2\), the horizontal optimal allocation is given by
\begin{equation*}
    m_i^*=\frac{\rho_i}{\rho_1+\rho_2}\bar{m}
\end{equation*}

\noindent That is, optimal meeting time and consequentially precision weighting \(h_i^*\) are proportional to the processing ability and so the DM can optimally weight their meeting time accordingly. However, for this means that there will always be some noise generated in the vertical structure. In contrast, when the communication ability differs \(\xi_1^2<\xi_2^2\) and worker one meets directly with the DM, then a sufficiently small DM meeting time \(\bar{m}\) would lead the DM to spend all their meeting time with worker 1 in the horizontal, effectively ignoring worker 2. Therefore, the information of worker 2 is not used in the horizontal structure. However, in the vertical structure for the same scarce meeting time, worker 1 is only constrained in their meeting with the DM and not worker 2. Hence, worker 1 can incorporate some of worker 2's information into their meeting with the DM. Thus, for a sufficiently scare amount of managerial attention, the vertical structure strictly improves information aggregation.   

Finally, for the symmetric processing ability cases I can extend the model to allow for $n$ workers. Figure \ref{fig:G central} shows that for the generalization in the horizontal structure, the DM now meets with $n$ workers and figure \ref{fig:G decentral} shows that in the vertical the DM still only meets with worker 1 but that they now meet with $n-1$ workers.

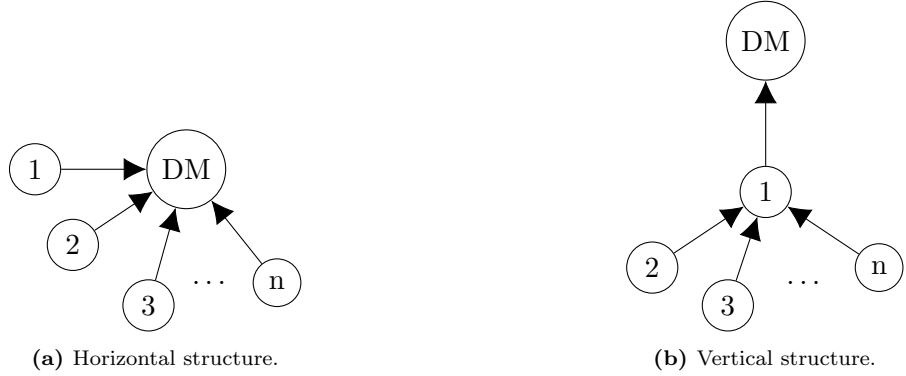
\begin{figure}[h]
\centering
    \begin{subfigure}[b]{0.45\textwidth}
        \centering
        \begin{tikzpicture}[node distance = 1cm, auto]
            \node [node] at (0,0) (DMa) {DM};
            \node [node] at (-2,0) (onea) {1};
            \node [node] at (-1.5,-1) (twoa) {2};
            \node [node] at (-.5,-1.8) (threea) {3};
            \node at (.3,-1.5)(a) {$\ldots$};
            \node [node] at (1.2,-1.5) (na) {n};
            \draw[->] (twoa) -- (DMa);
            \draw[->] (onea) -- (DMa);
            \draw[->] (threea) -- (DMa);
            \draw[->] (na) -- (DMa);
        \end{tikzpicture}
        \caption{Horizontal structure.} \label{fig:G central}
    \end{subfigure} \qquad 
    \begin{subfigure}[b]{0.45\textwidth}
        \centering
        \begin{tikzpicture}[node distance = 1cm, auto]
            \node [node] at (-1.5,-3) (two) {2};
            \node [node] at (0,-2) (one) {1};
            \node [node] at (0,0) (DM) {DM};
            \node [node] at (-0.5,-3.5) (three) {3};
            \node  at (0.5,-3.2) (b) {$\ldots$};
            \node [node] at (1.5,-3) (n) {n};
            \draw[->] (two) -- (one);
            \draw[->] (three) -- (one);
            \draw[->] (n) -- (one);
            \draw[->] (one) -- (DM);
        \end{tikzpicture}
        \caption{Vertical structure.} \label{fig:G decentral}
    \end{subfigure}
    
    \caption{Organizational Structures with $n$ Workers} \label{fig:general_structures}
\end{figure}

\noindent Introducing the extra workers in this way allows for this model to be applied to larger organizations or even teams within the organization. The horizontal structure can be thought of as a team where the decision-making is concentrated, in the sense that the DM has to gather all the information  dispersed within her team. In contrast, in the vertical structure, the DM delegates the information gathering activity to a single team member. As with the 2 worker version in this setting the results will rely on the heterogeneity in the workers communication skills. 
\begin{result}
\label{prop_general}
  	Assume symmetric processing ability, $\rho_i=\rho >0$ for all workers \(i\in{1,\dots,n}\) and \(n\geq2\).
    \begin{enumerate}
        \item If $\xi_i=\xi >0$ for all workers \(i\), then for every finite organizational capacity \(\kappa \in[\bar{m},\infty)\) the horizontal structure strictly dominates the vertical structure, the payoff difference converges to zero as \(\kappa \rightarrow \infty\)
        \item If $0<\xi_1<\xi_j=\xi $ for all \(j\geq 2\), let \(d\equiv\xi-\xi_1>0\). Define  
        \begin{equation*}
        \bar{\kappa}_n=
        \begin{cases}
        \bar{m}, & \bar{m}\leq \rho \xi_1 d,\\
        \kappa_{I,n}, & \bar{m}>\rho \xi_1 d \text{ and } \kappa_{I,n}\leq\kappa_{B,n},\\
        \kappa_{C,n}, & \bar{m}>\rho \xi_1  d \text{ and } \kappa_{I,n}>\kappa_{B,n}.
        \end{cases}
    \end{equation*}
    For a feasible \(\kappa \geq \bar{m}\), the vertical structure is strictly preferred if and only if \(\kappa > \bar{\kappa}_n\), at \(\kappa = \bar{\kappa}_n\) the structures have equal payoffs.
    \end{enumerate}

\end{result}
\noindent  The intuition for this is the same as for the 2 worker case. When the communication ability is homogeneous, then the DM is optimally weighting the underlying signals in the horizontal structure. But in the vertical she is over-weighting the signal from worker 1. This situation does not hold once worker 1 has better communicating ability than his co-workers and the workers have a relatively large amount of excess meeting capacity. is high enough. Now in the horizontal structure, the DM over-weights the signal of worker 1, since they are better able to communicate. However, in the vertical structure as the excess meeting capacity increases $k-\bar{m}$, the organization will reach a point where worker 1 can produce reports with the other workers so that the DM optimizes the weighting of the underlying signals. 

The $n$-worker case clarifies a fundamental trade-off regarding the managerial span of control. As the team size $n$ expands, direct consultation in the horizontal structure forces the decision-maker to divide her scarce attention budget $\bar{m}$ into increasingly smaller intervals ($\bar{m}/n$). Consequently, the communication noise in each direct meeting escalates rapidly, illustrating how severe attention bottlenecks cause flat communication structures to deteriorate as organizations scale.

The vertical structure alleviates this congestion at the top by allowing the decision-maker to dedicate her entire attention budget $\bar{m}$ to the most articulate team member. However, hierarchical delegation does not eliminate the informational burden of organizational size; rather, it reallocates the aggregation bottleneck to the intermediary. In the vertical hierarchy, worker 1 must divide her available meeting endowment $(\kappa - \bar{m})$ across the remaining $n-1$ subordinates, allocating $(\kappa - \bar{m})/(n-1)$ to each interaction. As $n \to \infty$, the noise in the synthesized report likewise compounds unless the organization possesses sufficient internal slack ($\kappa$ is sufficiently large). The critical threshold $\bar{\kappa}_n$ required for the hierarchy to outperform the flat network must therefore increase in $n$. The vertical structure proves superior not by expanding the absolute volume of attention, but by shifting the burden of information aggregation to an agent endowed with greater meeting capacity who can subsequently transmit the consolidated report to the decision-maker with maximal clarity

\section{Conclusion}
\label{conclusion}
In this paper, I analyze how a DM's relative time scarcity and workers' heterogeneous communication abilities shape the optimal organizational structure. By modeling the firm as a communication network under a team-theoretic framework, I compare the efficiency of horizontal and vertical structures in aggregating decentralized information. The clarity of the resulting reports depends on both the duration of meetings and the innate communication skills of the workers, subject to the organization's overall time constraints. 

I show that when the environment is symmetric with respect to the workers' processing and communication abilities, the horizontal structure is preferred to the vertical structure. This result holds even as the time available for workers to meet increases to very large values, which corresponds to a rise in the decision-maker's relative time scarcity. Under these symmetric conditions, the decision-maker optimally weights the signals that underlie each report. Furthermore, in the simple two worker model I show that the assumption of equal processing ability is not necessary. Finally, I extend the model to \(n\)-workers.

The findings of this paper yield several implications for managerial practice and the allocation of executive attention. First, while formal organizational structure could be rigid due to bureaucratic issues, the effective communication protocols
within firms can be managed flexibly. The model demonstrates that when executive attention is severely constrained by external demands, such as capital acquisition, regulatory proceedings, or enterprise sales, managers achieve higher decision quality by channeling internal reporting through a designated, highly articulate intermediary. In contrast, when executive time constraints are less binding, direct consultation with operational specialists becomes optimal
because unmediated communication eliminates the compounded transmission noise inherent in multi-layered reporting. The framework therefore rationalizes why executive meeting schedules and internal information routing—rather than formal corporate hierarchies—must adapt to shifts in managerial workload.

Second, this framework highlights a structural solution for managing technical experts who possess high processing ability but lack strong communication skills. Rather than risking the loss of their specialized knowledge in rushed, direct meetings with a time-constrained leader, firms can pair these experts with skilled communicators who act as information aggregators. This vertical arrangement ensures that valuable but poorly articulated signals are accurately synthesized and transmitted, preventing less-articulate workers from being inadvertently silenced by the friction of direct communication.

Finally, these theoretical predictions offer clear avenues for future empirical research. Utilizing digital trace data, such as internal email networks or calendar metadata, researchers could observe whether firms' internal communication structures steepen in response to exogenous shocks to executive workload, such as mergers and acquisitions or earnings seasons. Testing these dynamics empirically would further illuminate how modern organizations balance the fundamental trade-off between communication friction and the scarcity of managerial attention.

\appendix 
\section{Proofs} \label{appA}

This section details the proofs for the propositions and results in this paper. \\ 

\noindent \textbf{Proposition \ref{EquilActCentral}}
\begin{proof}
Given the observed reports $r_i$ for subordinate $i \in {1, 2}$ are linear and have normally distributed error terms this means that the $i$th report can be written as follows,  
\begin{equation*}
	\begin{split}
		r_i &= \theta  + \epsilon_i + \eta_i\\
     		      &= \theta+ \mathcal{N} \left( 0 , \frac{1}{\rho_i}\right) + \mathcal{N} \left( 0 , \frac{\xi_i^2}{m_i}\right) \\
       		      &= \theta+ \mathcal{N} \left( 0 , \frac{m_i+\xi_i^2\rho_i}{\rho_i m_i}\right) .
	\end{split}
\end{equation*} 
Where the precision weight $h_i$ is the inverse of the variance,
\begin{equation}
\label{eq:precision_weight}
	h_i=\frac{\rho_i m_i}{\xi_i^2\rho_i + m_i}, \hspace{10mm} \forall i=\{1,2\}.
\end{equation}
let $h_0\equiv\frac{1}{\sigma_\theta^2}$ denote the precision of the prior \(\theta \sim\mathcal{N}(\mu_\theta,\sigma_\theta^2)\). With Bayesian updating the posterior mean is equal to the precision weighted sum of the prior and sample observations over the total sum of the precision weights,
\begin{equation*}
\begin{split}
	\mathbb{E} \big[\theta | \mathbf{r} \big] =& \frac{h_o\mu_\theta+ \sum_{i=1}^2 h_i r_i }{h_o + \sum_{i=1}^2 h_i} \\
    =& \mu_\theta  + \omega_1 (r_1-  \mu_\theta ) + \omega_2(r_2-  \mu_\theta ) 
    \end{split}
\end{equation*}
Note that relative precision is given by
\begin{equation*}
 \omega_i \equiv \frac{h_i}{h_0+ \sum_{i=j}^2 h_j},
\end{equation*}
Here $\omega_i$ is the relative precision assigned by the DM to a particular report. Taking the diffuse prior limit $\sigma_\theta^2 \to \infty$ implies $h_0 \to 0$. In this limit, the relative prior weight vanishes ($\omega_0 \to 0$), satisfying $\sum_{i=1}^2 \omega_i = 1$, and the posterior mean simplifies directly to

\begin{equation*}
a^*(\mathbf{r}) = \mathbb{E}[\theta \mid \mathbf{r}] = \sum_{i=1}^2 \omega_i r_i, \quad \text{where} \quad \omega_i = \frac{h_i}{\sum_{j=1}^2 h_j} = \frac{\frac{\rho_i m_i}{\xi_i^2 \rho_i + m_i}}{\sum_{j=1}^2 \frac{\rho_j m_j}{\xi_j^2 \rho_j + m_j}}.
\end{equation*}
Substituting this optimal linear decision rule into the decision-maker's quadratic loss function gives

\begin{equation*}
\begin{split}
    \mathbb{E}\left[(\theta - a^*(\mathbf{r}))^2\right] &= \mathbb{E}\left[\left(\theta - \sum_{i=1}^2 \omega_i (\theta + \epsilon_i + \eta_i)\right)^2\right] \\
    &= \mathbb{E}\left[\left(\sum_{i=1}^2 \omega_i (\epsilon_i + \eta_i)\right)^2\right].
\end{split}
\end{equation*}
By the mutual independence of the noise terms across workers, the cross-product expectations vanish, leaving

\begin{equation*}
    \begin{split}
        \mathbb{E}\left[(\theta - a^*(\mathbf{r}))^2\right] &= \sum_{i=1}^2 \omega_i^2 \operatorname{Var}(\epsilon_i + \eta_i) \\
        &= \sum_{i=1}^2 \omega_i^2 \left(\frac{1}{h_i}\right) \\
        &= \sum_{i=1}^2 \left(\frac{h_i}{\sum_{j=1}^2 h_j}\right)^2 \frac{1}{h_i} \\ 
        &= \frac{\sum_{i=1}^2 h_i}{\left(\sum_{j=1}^2 h_j\right)^2} \\
        & = \frac{1}{\sum_{i=1}^2 h_i}.
    \end{split}
\end{equation*}
Consequently, the decision-maker's maximized expected utility is given by
\begin{equation}
\label{eq:DM_simple}
    U_H = \bar{u} - \frac{1}{\sum_{i=1}^2 h_i}.
\end{equation}

\end{proof} 

\noindent \textbf {Proposition \ref{EquilActDecent}}

\begin{proof}
Let worker 1 choose an unbiased linear aggregation rule \(M_1(s_1, r_{2,1}) = q s_1 + (1 - q) r_{2,1}\). The DM observes the noisy transmission
\begin{equation*}
r_{1,dm} = q(\theta + \epsilon_1) + (1 - q)(\theta + \epsilon_2 + \eta_2) + \eta_1 = \theta + q \epsilon_1 + (1 - q)(\epsilon_2 + \eta_2) + \eta_1.
\end{equation*}
Under a diffuse prior (\(\sigma_\theta^2 \to \infty\)), the DM's optimal Bayesian action given the single report \(r_{1,dm}\) is \(a(r_{1,dm}) = \mathbb{E}[\theta \mid r_{1,dm}] = r_{1,dm}\). Substituting \(a(r_{1,dm})\) into the organizational loss function yields 
\begin{equation*}
\mathbb{E}\left[(\theta - a)^2 \mid q\right] = \mathbb{E}\left[\left(q \epsilon_1 + (1 - q)(\epsilon_2 + \eta_2) + \eta_1\right)^2\right]
\end{equation*}
By the independence of 
\(\epsilon_1, \epsilon_2, \eta_2,\) and \(\eta_1\), this simplifies to: 
\begin{equation*}
    \begin{split}
        \mathbb{E}\left[(\theta - a)^2 \mid q\right] &= q^2 \text{Var}(\epsilon_1) + (1 - q)^2 \left(\text{Var}(\epsilon_2) + \text{Var}(\eta_2)\right) + \text{Var}(\eta_1) \\
        &=\frac{q^2}{\rho_1} + (1 - q)^2 \left(\frac{1}{\rho_2} + \frac{\xi_2^2}{m_2}\right) + \frac{\xi_1^2}{m_1} 
    \end{split}
\end{equation*}
Let 
\begin{equation*}
    h_{1,2} \equiv \left(\frac{1}{\rho_2} + \frac{\xi_2^2}{m_2}\right)^{-1} = \frac{m_2 \rho_2}{m_2 + \xi_2^2 \rho_2}.
\end{equation*}
The team objective is strictly convex in \(q\). Differentiating with respect to \(q\) and setting the first-order conditions to zero gives
\begin{equation*}
    \frac{2q}{\rho_1} - \frac{2(1 - q)}{h_{1,2}} = 0 \implies q^* = \frac{\rho_1}{\rho_1 + h_{1,2}} \equiv q_{1,1},
\end{equation*}
and
\begin{equation*}
    1 - q^* = \frac{h_{1,2}}{\rho_1 + h_{1,2}} \equiv q_{1,2}.
\end{equation*}
Substituting \(q^*\) back into the loss expression yields the minimized variance
\begin{equation*}
\begin{split}
        \mathbb{E}\left[(\theta - a)^2\right] &= \left(\frac{\rho_1}{\rho_1 + h_{1,2}}\right)^2 \frac{1}{\rho_1} + \left(\frac{h_{1,2}}{\rho_1 + h_{1,2}}\right)^2 \frac{1}{h_{1,2}} + \frac{\xi_1^2}{m_1} \\
        &= \frac{1}{\rho_1 + h_{1,2}} + \frac{\xi_1^2}{m_1},
\end{split}
\end{equation*}
which gives the expected utility 
\begin{equation*}
    U_V = \bar{u} - \left(\frac{1}{\rho_1 + \frac{m_2 \rho_2}{m_2 + \xi_2^2 \rho_2}} + \frac{\xi_1^2}{m_1}\right)
\end{equation*}
Hence the Bayesian conditional expectation \(M_1^*(s_1, r_{2,1}) = \mathbb{E}[\theta \mid s_1, r_{2,1}]\) is the unique optimal reporting rule. \\

\end{proof}

\noindent \textbf{Result \ref{prop_same_xi}}
\begin{proof}
Firstly, I will determine the optimal meeting time for the horizontal and vertical structures.\\
 
\noindent \textbf{Horizontal Structure.} To determine the optimal allocation of meeting time the DM maximises their expected utility found in (\ref{eq:DM_simple}) subject to the meeting time constraint $\sum_{i=1}^2m_i\leq \bar{m}$ since the DM's constraint is the one that binds. As the utility increases in $m_i$, the meeting time is exhausted $\sum_{i=1}^2m_i=\bar{m}$. For now, I will keep the communication and processing ability heterogeneous. This is equivalent to the following maximization problem
\begin{equation*}
\begin{aligned}
& \max_{\textbf{m} \in \mathbb{R}}
&&  \sum_{i=1}^2 h_i  \\
& \text{subject to}
&& \sum_i^2 m_i =\bar{m} .\\
\end{aligned}
\end{equation*}
Plugging in the constraint into the objective function to get it in terms of $m_1$ gives
\begin{equation*}
	V(m_1)=\frac{m_1\rho_1}{m_1+\rho_1\xi_1^2} +\frac{(\bar{m}-m_1)\rho_2}{\bar{m}-m_1+\xi_2^2\rho_2}.
\end{equation*} 
The function is strictly concave by the second-order condition
\begin{equation*}
	\pdv[2]{V(m_1)}{m_1}=-2 \left[\frac{\xi_1^2\rho_1^2}{(\xi_1^2\rho_1+m_1)^3} +\frac{\xi_2^2\rho_2^2}{(\bar{m}-m_1+\xi_2^2\rho_2)^3} \right] <0 \hspace{3mm} \forall\ m_1,
\end{equation*}
this means that any critical point found in the following maximization problem will be at least a local maximum and there is at most one global maximum. The Lagrangian is the following
\begin{equation*}
\mathcal{L} (\textbf{m}) =  \sum_{i=1}^2 h_i - \lambda(\sum_i^2 m_i -\bar{m} ).
\end{equation*}
The first-order condition with respect to $m_i $ is as follows
\begin{equation}
\label{eq3}
\pdv{\mathcal{L}}{m_i}=0 \Leftrightarrow   \pdv{h_i}{m_i}  = \lambda.
\end{equation}   
Where 
\begin{equation*}
\begin{split}
\pdv{h_i}{m_i}  &=\frac{\xi_i^2\rho_i^2}{(\xi_i^2\rho_i+m_i)^2}\\
&=\frac{m_i^2\xi_i^2\rho_i^2}{m_i^2(\xi_i^2\rho_i+m_i)^2}\\
&=\frac{h_i^2\xi_i^2}{m_i^2}.\\
\end{split}
\end{equation*}

Combining this with (\ref{eq3})
\begin{equation*}
\begin{split}
\frac{h_i^2\xi_i^2}{m_i^2}  &= \lambda \\
m_i &= \frac{\xi_i h_i}{\sqrt{\lambda}}  \\
m_i &= \frac{\xi_i\rho_i m_i}{(\xi_i^2\rho_i+m_i)\sqrt{\lambda}}  \\
m_i\sqrt{\lambda} &= \xi_i\rho_i -\xi_i^2\rho_i\sqrt{\lambda}  \\
\end{split}
\end{equation*}
this gives
\begin{equation}
\label{optimal_meeting_cent_lambda}
    m_i= \frac{\xi_i\rho_i(1-\xi_i\sqrt{\lambda})}{\sqrt{\lambda}} .
\end{equation}

Using equation (\ref{optimal_meeting_cent_lambda}) and $\sum_{i=1}^2 m_i=\bar{m}$ to solve for $\lambda$ in terms of the primitives of the model gives
\begin{equation*}
\begin{split}
m_i&= \frac{\xi_i\rho_i(1-\xi_i\sqrt{\lambda})}{\sqrt{\lambda}}  \\
\sum_{i=1}^2 m_i&= \frac{\sum_{i=1}^2\xi_i\rho_i(1-\xi_i\sqrt{\lambda})}{\sqrt{\lambda}}  \\
\bar{m}&= \frac{\sum_{i=1}^2\xi_i\rho_i(1-\xi_i\sqrt{\lambda})}{\sqrt{\lambda}}  \\
\sqrt{\lambda}  &= \frac{\sum_{i=1}^2\xi_i\rho_i}{\bar{m}  + \sum_{i=1}^2\xi_i^2\rho_i } . \\
\end{split}
\end{equation*}
Now plugging $\sqrt{\lambda}$ back into (\ref{optimal_meeting_cent_lambda}) to yield
\begin{equation*}
\begin{split}
m_i^*&= \frac{\xi_i\rho_i}{\sqrt{\lambda}}(1-\xi_i\sqrt{\lambda})  \\
  &= \xi_i\rho_i \frac{\left(\bar{m}  + \sum_{i=1}^2\xi_i^2\rho_i \right)}{\sum_{i=1}^2\xi_i\rho_i } \left(1-\frac{\xi_i\sum_{i=1}^2\xi_i\rho_i}{\bar{m}  + \sum_{i=1}^2\xi_i^2\rho_i} \right)  \\
  &= \xi_i\rho_i \frac{\left(\bar{m}  + \sum_{i=1}^2\xi_i^2\rho_i \right)}{\sum_{i=1}^2\xi_i\rho_i } \left(\frac{\bar{m}  + \sum_{i=1}^2\xi_i^2\rho_i - \xi_i\sum_{i=1}^2\xi_i\rho_i}{\bar{m}  + \sum_{i=1}^2\xi_i^2\rho_i } \right)  \\
    &= \xi_i\rho_i \frac{\left(\bar{m}+ \sum_{i=1}^2\xi_i^2\rho_i - \xi_i\sum_{i=1}^2\xi_i\rho_i \right)}{\sum_{i=1}^2\xi_i\rho_i } . \\
\end{split}
\end{equation*}
The optimal allocation of meeting time is thus
\begin{equation}
\label{optimal_meeting_cent}
m_i^*= \xi_i\rho_i \frac{\left(\bar{m}+ \sum_{i=1}^2\xi_i^2\rho_i - \xi_i\sum_{i=1}^2\xi_i\rho_i \right)}{\sum_{i=1}^2\xi_i\rho_i }  .
\end{equation}
Given that $h_i=\frac{\rho_i m_i}{\xi_i^2\rho_i +m_i}$ this can be expressed in terms of the primitives of the model
\begin{equation*}
\begin{split}
h_i^*&=\frac{\xi_i\rho_i^2}{\xi_i^2\rho_i + \xi_i\rho_i \frac{\left(\bar{m}+ \sum_{i=1}^2\xi_i^2\rho_i - \xi_i\sum_{i=1}^2\xi_i\rho_i \right)}{\sum_{i=1}^2\xi_i\rho_i }}  \frac{\left(\bar{m}+ \sum_{i=1}^2\xi_i^2\rho_i - \xi_i\sum_{i=1}^2\xi_i\rho_i \right)}{\sum_{i=1}^2\xi_i\rho_i }\\
    &= \frac{\xi_i\rho_i^2}{\xi_i^2\rho_i\sum_{i=1}^2\xi_i\rho_i  + \xi_i\rho_i \left(\bar{m}+ \sum_{i=1}^2\xi_i^2\rho_i - \xi_i\sum_{i=1}^2\xi_i\rho_i \right)} \left(\bar{m}+ \sum_{i=1}^2\xi_i^2\rho_i - \xi_i\sum_{i=1}^2\xi_i\rho_i \right)\\
	&= \frac{\xi_i\rho_i^2}{\xi_i^2\rho_i\sum_{i=1}^2\xi_i\rho_i  +  \xi_i\rho_i\bar{m}+ \xi_i\rho_i\sum_{i=1}^2\xi_i^2\rho_i -\xi_i^2\rho_i \sum_{i=1}^2\xi_i\rho_i } \left(\bar{m}+ \sum_{i=1}^2\xi_i^2\rho_i - \xi_i\sum_{i=1}^2\xi_i\rho_i \right)\\
	&= \frac{\rho_i\left(\bar{m}+ \sum_{i=1}^2\xi_i^2\rho_i - \xi_i\sum_{i=1}^2\xi_i\rho_i \right)}{  \bar{m}+ \sum_{i=1}^2\xi_i^2\rho_i } .
\end{split}
\end{equation*}
Given that processing and communicating ability are assumed to be symmetric, $\rho_1=\rho_2=\rho$ and $\xi_1^2=\xi_2^2=\xi^2$ plugging this into (\ref{optimal_meeting_cent}) yields
\begin{equation*}
\begin{split}
m_i^*&= \xi \rho \frac{\left(\bar{m}+ \sum_{i=1}^2\xi^2\rho - \xi\sum_{i=1}^2\xi\rho \right)}{\sum_{i=1}^2\xi\rho }  \\
	&= \xi \rho \frac{\left(\bar{m}+ \xi^2\rho \sum_{i=1}^2- \xi^2\rho \sum_{i=1}^2 \right)}{\xi\rho \sum_{i=1}^2 }  \\
	&=  \frac{\bar{m} }{2  } . \\
\end{split}
\end{equation*}
The DM optimally spends an equal amount of time meeting both workers. This leads to equal precision weights $h_1^*=h_2^*$ and so the DM equally weights both reports $\omega_1^*=\omega_2^*=1/2$ .

\noindent Let the optimized horizontal payoff be denoted by \(U_H\), from using \(m^*=\bar{m}/2\) and Proposition \ref{EquilActCentral} this gives

\begin{equation}
\label{eq:Horizontal_payoff_symmetric}
U_H= \bar{u} - \frac{1}{2\rho} - \frac{\xi^2}{\bar{m}}
\end{equation}

\noindent\textbf{Vertical structure.}
For any feasible vertical allocation,
\begin{equation*}
m_1\leq\bar m,\qquad m_1+m_2\leq\kappa,
\end{equation*}
the payoff is
\begin{equation*}
U_V(m_1,m_2) =\bar u-\frac{1}{\rho+h_2(m_2)}-\frac{\xi^2}{m_1}, \qquad h_2(m_2)=\frac{\rho m_2}{m_2+\rho\xi^2}.
\end{equation*}
For finite $m_2$, $h_2(m_2)<\rho$, and $m_1\leq\bar m$. Hence
\begin{equation*}
U_V(m_1,m_2) < \bar u-\frac{1}{2\rho}-\frac{\xi^2}{\bar m} =U_H.
\end{equation*}
Therefore the horizontal structure strictly dominates for every
finite $\kappa\geq\bar m$. \\

\noindent To show convergence, consider the feasible vertical allocation
$m_1=\bar m,\ m_2=\kappa-\bar m$. Its payoff converges to
$U_H$ as $\kappa\to\infty$. Since the optimized vertical payoff is bounded above by $U_H$ and bounded below by the payoff from this feasible allocation, the optimized vertical payoff also converges to $U_H$.
\end{proof}

\noindent \textbf{Result \ref{prop_diff_xi}}
\begin{proof}
Assume $\rho_1=\rho_2=\rho>0$ and $0<\xi_1<\xi_2$. Let
$d=\xi_2-\xi_1>0$. \\

\noindent\textbf{Horizontal structure.}
The horizontal structure maximizes
\begin{equation}
H_T=h_1(m_1)+h_2(m_2) \quad\text{subject to}\quad m_1+m_2\leq\bar m,
\end{equation}
where
\begin{equation}
h_i(m_i)=\frac{\rho m_i}{m_i+\rho\xi_i^2}.
\end{equation}
The time constraint binds. At an interior optimum, equality of marginal precisions gives
\begin{equation}
\frac{\xi_1}{m_1+\rho\xi_1^2} = \frac{\xi_2}{m_2+\rho\xi_2^2}.
\end{equation}
Solving this condition together with $m_1+m_2=\bar m$ gives
\begin{equation}
m_1^H= \frac{\xi_1\bigl[\bar m+\rho\xi_2d\bigr]}{\xi_1+\xi_2}, \qquad m_2^H= \frac{\xi_2\bigl[\bar m-\rho\xi_1d\bigr]}{\xi_1+\xi_2}.
\end{equation}
Thus the interior allocation applies when $\bar m>\rho\xi_1d$. If $\bar m\leq\rho\xi_1d$, the constrained optimum is the corner
$m_1^H=\bar m,\ m_2^H=0$.\\

\noindent\emph{Case I: Horizontal corner region.}
Suppose $\bar m\leq\rho\xi_1d$. The horizontal payoff is
\begin{equation}
U_H^C=\bar u-\frac{1}{\rho}-\frac{\xi_1^2}{\bar m}.
\end{equation}

In the vertical structure, the time constraints are
\begin{equation}
m_1\leq\bar m,\qquad m_1+m_2\leq\kappa.
\end{equation}
For $\kappa\geq\bar m$, the total time constraint binds, so $m_2=\kappa-m_1$. The vertical payoff as a function of $m_1$ is
\begin{equation}
U_V(m_1)= \bar u- \frac{\kappa-m_1+\rho\xi_2^2} {\rho\left[2(\kappa-m_1)+\rho\xi_2^2\right]} -\frac{\xi_1^2}{m_1}.
\end{equation}
Its unconstrained optimum is
\begin{equation}
\widehat m_1(\kappa) =\frac{\xi_1(2\kappa+\rho\xi_2^2)}{2\xi_1+\xi_2},
\end{equation}
and the constrained optimum is
\begin{equation}
m_1^*(\kappa)=\min\{\bar m,\widehat m_1(\kappa)\}, \qquad m_2^*(\kappa)=\kappa-m_1^*(\kappa).
\end{equation}
The cap binds for all $\kappa\geq\bar m$ in this case: indeed, $\bar m\leq\rho\xi_1d<\rho\xi_1\xi_2$ implies $\widehat m_1(\bar m)>\bar m$, and $\widehat m_1(\kappa)$ increases in $\kappa$. Hence $m_1^*=\bar m$ and $m_2^*=\kappa-\bar m$.
At $\kappa=\bar m$, the two structures have equal payoffs. For every $\kappa>\bar m$, $m_2^*>0$, so the vertical structure incorporates worker 2's information and strictly improves on the horizontal corner. Thus the threshold in this case is $\bar\kappa=\bar m$. \\

\noindent\emph{Case II: Horizontal interior region.}
Suppose $\bar m>\rho\xi_1d$. Define
\begin{equation}
S=\xi_1+\xi_2,\qquad D_H=2\bar m+\rho d^2,\qquad B=2\xi_1+\xi_2,\qquad \psi=\rho\xi_2^2.
\end{equation}
Using the interior horizontal allocation,
\begin{equation}
U_H^I =\bar u-\frac{\bar m+\rho(\xi_1^2+\xi_2^2)}
{\rho\left(2\bar m+\rho d^2\right)} =\bar u-\frac{1}{2\rho}-\frac{S^2}{2D_H}.
\end{equation}
For the vertical structure, since $m_2=\kappa-m_1$, the loss is
\begin{equation}
L_V(m_1)= \frac{\kappa-m_1+\rho\xi_2^2}{\rho\left[2(\kappa-m_1)+\rho\xi_2^2\right]} +\frac{\xi_1^2}{m_1}.
\end{equation}
It is strictly convex in $m_1$. Its first-order condition gives
\begin{equation}
\widehat m_1(\kappa) =\frac{\xi_1(2\kappa+\psi)}{B}, \qquad m_1^*(\kappa)=\min\{\bar m,\widehat m_1(\kappa)\}.
\end{equation}
The cap starts binding at
\begin{equation}
\kappa_B =\frac{\bar m B}{2\xi_1}-\frac{\psi}{2}.
\end{equation}
When $\kappa\leq\kappa_B$, the vertical optimum is interior in $m_1$, and substitution of the first-order condition gives
\begin{equation}
U_V^I(\kappa) =\bar u-\frac{1}{2\rho} -\frac{B^2}{2(2\kappa+\psi)}.
\end{equation}
The crossing of this branch with $U_H^I$ occurs at
\begin{equation}
\kappa_I=\frac{D_H(B/S)^2-\psi}{2}.
\end{equation}

When $\kappa\geq\kappa_B$, the cap binds, $m_1^*=\bar m$. Writing $t=\kappa-\bar m$, the vertical payoff is
\begin{equation}
U_V^C(t)= \bar u-\frac{t+\psi}{\rho(2t+\psi)} -\frac{\xi_1^2}{\bar m}.
\end{equation}
Define
\begin{equation}
F=\bar m(3\xi_1+\xi_2)-2\rho\xi_1^2d.
\end{equation}
Since $\bar m>\rho\xi_1d$,
\begin{equation}
F>\rho\xi_1d(\xi_1+\xi_2)>0.
\end{equation}
The crossing of this capped branch with $U_H^I$ occurs at
\begin{equation}
\kappa_C=\bar m+\frac{\xi_2^2(\bar m-\rho\xi_1d)^2}{Fd}.
\end{equation}
At $\kappa=\bar m$, the vertical payoff is strictly below $U_H^I$.
If the cap binds there, the difference is
\begin{equation}
U_H^I-U_V^C =\frac{(\bar m-\rho\xi_1d)^2}
{\rho\bar mD_H}>0.
\end{equation}
If the vertical optimum is interior there, then $\bar m>\rho\xi_1\xi_2$. Setting $q=\bar m-\rho\xi_1\xi_2>0$, direct substitution gives
\begin{equation}
U_H^I-U_V^I(\bar m) =
\frac{\xi_1\left[q(3\xi_1+2\xi_2)+\rho\xi_1^2(2\xi_1+\xi_2)
\right]}{(2\bar m+\rho\xi_2^2)D_H}>0.
\end{equation}

As $\kappa\to\infty$, the cap binds and
\begin{equation}
\lim_{\kappa\to\infty}U_V^C =\bar u-\frac{1}{2\rho}-\frac{\xi_1^2}{\bar m}.
\end{equation}
The limiting payoff exceeds $U_H^I$, since
\begin{equation}
\lim_{\kappa\to\infty}U_V^C-U_H^I=\frac{dF}{2\bar mD_H}>0.
\end{equation}
The optimized vertical payoff is continuous and strictly increasing in $\kappa$: increasing $\kappa$ expands the feasible set and permits a strict increase in $m_2$ at any fixed feasible $m_1$. Therefore there is exactly one crossing. It lies on the interior branch if $\kappa_I\leq\kappa_B$, and on the capped branch otherwise. Hence
\begin{equation}
\bar\kappa=
\begin{cases}
\kappa_I,& \kappa_I\leq\kappa_B,\\
\kappa_C,& \kappa_I>\kappa_B.
\end{cases}
\end{equation}
The vertical structure strictly dominates the horizontal structure if and only if $\kappa>\bar\kappa$.
\end{proof}

\noindent \textbf{Result \ref{prop:asymmetric_rhos}}
\begin{proof}
The first section will correspond to homogeneous communication and the second section will correspond to heterogeneous communication.
\subsubsection*{Case I: $\xi_i^2=\xi^2 $}
Let

\begin{equation*}
t=\kappa-\bar{m},
\end{equation*}

\noindent Let $U_H$ denote the optimized Horizontal payoff and $U_{V_1}$ the optimized Vertical payoff when worker $1$ is the intermediary. After terms common to both structures are canceled out, the payoff difference when worker 1 is the intermediary is

\begin{equation*}
\begin{aligned}
U_H-U_{V_1}
&=
\frac{t+\xi^2\rho_2}
{t(\rho_1+\rho_2)+\xi^2\rho_1\rho_2}
-
\frac{1}{\rho_1+\rho_2}
\\
&=
\frac{\xi^2\rho_2^2}
{(\rho_1+\rho_2)
\left[
t(\rho_1+\rho_2)+\xi^2\rho_1\rho_2
\right]}
>0.
\end{aligned}
\end{equation*}

\noindent Hence, for $\rho_1,\rho_2,\xi>0$ and $0<\bar{m}\leq\kappa<\infty$, the Horizontal structure strictly dominates either Vertical structure. Moreover,

\begin{equation*}
U_H-U_{V_j}\longrightarrow 0
\qquad\text{as}\qquad
\kappa\longrightarrow\infty,
\end{equation*}

\subsubsection*{Case I: $\xi_1^2<\xi_2^2 $}

\emph{Corner meeting allocation} Suppose worker 1 is the better communicator, so that $0<\xi_1<\xi_2$. Under the Horizontal structure, the DM maximizes $U_H(m_1)$ defined above. Since $U_H$ is strictly concave, the corner allocation
\begin{equation*}
m_1^*=\bar{m}, \qquad m_2^*=0
\end{equation*}
is optimal if and only if
\begin{equation*}
U_H'(\bar{m})\geq 0.
\end{equation*}
Evaluating the derivative at $m_1=\bar{m}$, this condition is equivalent to
\begin{equation*}
\bar{m} \leq \rho_1\xi_1(\xi_2-\xi_1).
\end{equation*}
Thus, when managerial meeting time is sufficiently scarce, the DM allocates all of her time to worker 1 and does not directly use worker 2's information. The optimized Horizontal payoff is then
\begin{equation*}
U_H = \bar{u} - \frac{1}{\rho_1} - \frac{\xi_1^2}{\bar{m}}.
\end{equation*}
In the Vertical structure, worker 1 can use $\kappa-\bar{m}$ units of time to collect information from worker 2 before communicating with the DM. The Vertical payoff is
\begin{equation*}
U_V = \bar{u} -
\frac{1}{ \rho_1+ \frac{(\kappa-\bar{m})\rho_2} {\kappa-\bar{m}+\xi_2^2\rho_2} } - \frac{\xi_1^2}{\bar{m}}.
\end{equation*}
Let
\begin{equation*}
a = \frac{(\kappa-\bar{m})\rho_2} {\kappa-\bar{m}+\xi_2^2\rho_2}.
\end{equation*}
For $\kappa>\bar{m}$ and $\rho_2>0$, $a>0$. Therefore,
\begin{equation*}
\begin{split}
U_V-U_H
&= \frac{1}{\rho_1} - \frac{1}{\rho_1+a}
\\
&= \frac{a}
{\rho_1(\rho_1+a)} >0.
\end{split}
\end{equation*}
\emph{Interior meeting allocation} When $\bar{m} > \rho_1 \xi_1 (\xi_2 - \xi_1)$, the horizontal allocation is interior, generating the aggregated precision 
\begin{equation*}
    H_T(\bar{m}) = \frac{(\rho_1 + \rho_2)\bar{m} + \rho_1 \rho_2 (\xi_2 - \xi_1)^2}{\bar{m} + \rho_1 \xi_1^2 + \rho_2 \xi_2^2}
\end{equation*}
and payoff 
\begin{equation*}
    U_H^I = \bar{u} - \frac{1}{H_T(\bar{m})}.
\end{equation*}
At the lower boundary $\kappa = \bar{m}$, intermediate transmission is impossible ($t = 0$), so 
\begin{equation*}
    U_V(\bar{m}) = \bar{u} - \frac{1}{\rho_1} - \frac{\xi_1^2}{\bar{m}} < U_H^I
\end{equation*}
In the asymptotic limit as $\kappa \to \infty$, intermediate communication noise vanishes, yielding 
\begin{equation*}
   \lim_{\kappa \to \infty} U_V(\kappa) = \bar{u} - \frac{1}{\rho_1 + \rho_2} - \frac{\xi_1^2}{\bar{m}} > U_H^I.
\end{equation*}
Because $U_V(\kappa)$ is strictly increasing and continuous in $\kappa$, the intermediate value theorem guarantees the existence of a unique crossing point $\bar{\kappa} > \bar{m}$ satisfying $U_V(\bar{\kappa}) = U_H^I$, above which the vertical structure is strictly preferred. If the vertical cap binds at the crossing, equating the capped vertical payoff to the horizontal interior payoff gives
\begin{equation*}
\bar\kappa = \bar m+ \frac{\rho_2\xi_2^2(1-\rho_1\Omega)}{(\rho_1+\rho_2)\Omega-1}.
\end{equation*}

\begin{equation*}
\Omega \equiv \frac{\bar m+\rho_1\xi_1^2+\rho_2\xi_2^2} {(\rho_1+\rho_2)\bar m+\rho_1\rho_2(\xi_2-\xi_1)^2} -\frac{\xi_1^2}{\bar m}.
\end{equation*}

\end{proof}

\noindent \textbf{Result \ref{prop_general}}
\begin{proof}
The first section will correspond to homogeneous communication and the second section will correspond to heterogeneous communication.\\

\noindent For a worker with processing ability \(\rho\) and communication ability \(\xi\), define the precision weight of a meeting of length \(m\) as
\begin{equation*}
    h(m;\xi)=\frac{\rho m}{m+\rho\xi^2},
\end{equation*}
\noindent with first and second derivatives with respect to \(m\) given by 
\begin{equation}
\label{eq:precision_derivatove}
    h^\prime(m;\xi)=\frac{\rho^2 \xi^2}{(m+\rho\xi^2)^2}, \qquad  h^{\prime\prime}(m;\xi)=-2\frac{\rho^2 \xi^2}{(m+\rho\xi^2)^3}<0.
\end{equation}
Let
\begin{equation}
\label{def:worker1_worker_j_comm_diff}
    d=\xi-\xi_1
\end{equation}

\subsubsection*{Case I: $\xi_i^2=\xi^2 $}

\noindent Performing the same steps as I did for the proof of proposition \ref{EquilActCentral} but with $n$ workers will lead to the following expected utility. Where $h_i$ corresponds to a precision weight
\begin{equation}
\label{eq2}
U_H=\bar{u}-\mathbb{E} \left[ (\theta -a^*)^2 \right]=\bar{u}-  \frac{1}{\sum_{i=1}^n h_i}. 
\end{equation}
The higher the precision of each report, the smaller the quadratic loss and the better the performance of the firm in estimating $\theta$. Given that processing and communicating ability are assumed to be symmetric, it will be optimal for the DM to meet each worker for the same amount of time
\begin{equation*}
m_i^*=\frac{\bar{m}}{n}\ \forall\ i .
\end{equation*}
This leads to equal precision weights $h_i^*=h$ and so the total precision is given by
\begin{equation*}
    \begin{split}
        \sum_{i=1}^n h &= n h \\
        =& n \frac{\rho \frac{\bar{m}}{n}}{ \frac{\bar{m}}{n}+ \rho \xi^2} \\
        =&  \frac{n\rho\bar{m}}{\bar{m}+n\rho \xi^2},
    \end{split}
\end{equation*}
\noindent giving
\begin{equation*}
    U_H= \bar{u} - \frac{1}{n\rho} - \frac{\xi^2}{\bar{m}}.
\end{equation*}

\noindent For any vertical allocation the report between worker 1 and the DM will be
\begin{equation*}
	\begin{split}
 			r_{1,dm} &=\mathbb{E}\left[\theta| \{s_1, r_{2,1}, r_{3,1},\ldots, r_{n,1} \}\right] + \eta_1 \\
                			&=  \theta  + q_{1,1}\epsilon_1+\sum_{i=2}^n q_{1,i}\left( \epsilon_i+\eta_i \right)+ \eta_1 .
	\end{split}
\end{equation*} 
\noindent This means that, given the independence of the error terms, the objective function can be re-arranged as follows
\begin{equation}
\label{eq:vertical_n_workers_util}
	\begin{split}
			\mathbb{E} \left[ (\theta -a^*)^2 \right] &=\mathbb{E} \left[ (\theta -r_{1,dm} )^2 \right] \\
						            &=\left[ \frac{\rho}{\rho+\sum\limits_{i=2}^n \frac{m_i \rho}{m_i +\xi^2\rho} }\right]^2 \frac{1}{\rho}+\left[ \frac{\frac{m_2 \rho}{m_2 +\xi^2\rho} }{\rho+\sum\limits_{i=2}^n \frac{m_i \rho}{m_i +\xi^2\rho} }\right]^2\left( \frac{m_2 +\xi^2\rho}{\rho m_2}\right)+ \\
							& \ldots+\left[ \frac{\frac{m_n \rho}{m_n +\xi^2\rho} }{\rho+\sum\limits_{i=2}^n \frac{m_i \rho}{m_i +\xi^2\rho}}\right]^2\left( \frac{m_n +\xi^2\rho}{\rho m_n}\right)+ \frac{\xi^2}{m_1} \\
						     &= \frac{1}{\rho+\sum\limits_{i=2}^n \frac{m_i \rho}{m_i +\xi^2\rho} }+\frac{\xi^2}{m_1} \\
                              &= \frac{1}{\rho+\sum\limits_{i=2}^n h(m_i;\xi) }+\frac{\xi^2}{m_1} .\\
	\end{split}
\end{equation}

\noindent Define worker 1's precision before they meet the DM
\begin{equation*}
    G_V = \rho+\sum\limits_{i=2}^n h(m_i;\xi).
\end{equation*}
\noindent Since \(h(m_i;\xi)<\rho\) for every finite \(m_i\),
\begin{equation*}
    G_V < n \rho.
\end{equation*}
Moreover, worker 1's meeting with the DM must satisfy \(m_1 \leq\bar{m}\), which gives
\begin{equation*}
    \begin{split}
        U_V &= \bar{u}- \frac{1}{G_V} - \frac{\xi^2}{m_1} \\
            &\leq \bar{u}- \frac{1}{n\rho} - \frac{\xi^2}{\bar{m}} = U_H
    \end{split}
\end{equation*}
Given \(G_V < n \rho\) if  \(m_1<\bar{m}\Rightarrow -\frac{1}{\bar{m}}>-\frac{1}{m_1}\) then \(U_H-U_V>0\); if \(m_1=\bar{m}\), then given \(m_2 \leq \kappa-\bar{m} <\infty\ \text{and}\ G_V < n \rho\), gives \(U_H-U_V>0\). Hence for finite \(\kappa\) the horizontal structure dominates. When \(m_1=\bar{m}\) and distributing the remaining meeting time \(\kappa-\bar{m}\) across the other meetings it is clear as \(\kappa\rightarrow \infty\) that \(G_V \rightarrow n\rho\) and the payoffs of the two structures converge.

\subsubsection*{Case II: $\xi_1^2 < \xi^2 \ \forall i \neq 1$ }
\emph{Corner solution.} Consider the horizontal structure,

\begin{equation*}
\begin{aligned}
& \max_{m_i \geq 0}
&&   h(m_1;\xi_1) + \sum_{j=2}^nh(m_j;\xi)  \\
& \text{subject to}
&& \sum_i^n m_i =\bar{m} .\\
\end{aligned}
\end{equation*}
The objective is strictly concave as the second order can be shown to be negative using (\ref{eq:precision_derivatove}). Suppose \(\bar{m}\) is sufficiently small, the corner solution \(m_1=\bar{m}, m_j\geq 2\) is optimal if and only if 
\begin{equation*}
    h^\prime(\bar{m};\xi_1)\geq  h^\prime(0;\xi),
\end{equation*}
that is the marginal return from giving an extra unit to worker 1 is at least as large as starting a meeting with any other worker. substituting \(h^\prime\) and using (\ref{def:worker1_worker_j_comm_diff}) gives
\begin{equation*}
    \begin{split}
         \frac{\rho^2\xi_1^2}{(\bar{m}+\rho \xi_1^2)^2} \geq\frac{1}{\xi^2} &\Leftrightarrow \frac{\rho\xi_1}{(\bar{m}+\rho \xi_1^2)} \geq\frac{1}{\xi} \\
         &\Leftrightarrow \rho\xi_1 \xi \geq (\bar{m}+\rho \xi_1^2) \\
          &\Leftrightarrow \bar{m} \leq \rho\xi_1(\xi - \xi_1) = \rho\xi_1 d
    \end{split}
\end{equation*}
At this corner the horizontal payoff is
\begin{equation}
    \label{eq:horizontal_payoff_corner_n_workers}
    U_H^C = \bar{u} - \frac{1}{\rho} - \frac{\xi_1^2}{\bar{m}}
\end{equation}
Suppose \(\kappa>\bar{m}\) and let \(t=\kappa-\bar{m}>0\). Consider a feasible allocation where worker 1 spends all DM's meeting capacity with the DM and divides \(t\) equally among the other workers \(2,\dots,n\) 
\begin{equation*}
    m_1=\bar{m}, \qquad m_j = \frac{t}{n-1}\ j\geq 2.
\end{equation*}
As can be seen from (\ref{eq:vertical_n_workers_util}) each of these workers contributes precision
\begin{equation*}
    h\big(\frac{t}{n-1};\xi \big)= \frac{t\rho}{t + \rho\xi^2(n-1)}.
\end{equation*}
So the total contribution of these \(n-1\) workers is given by
\begin{equation*}
    H = \frac{t\rho(n-1)}{t + \rho\xi^2(n-1)}>0.
\end{equation*}
The feasible vertical payoff is given by
\begin{equation*}
    \tilde{U_V} = \bar{u} - \frac{1}{\rho + H} - \frac{\xi_1^2}{\bar{m}},
\end{equation*}
and from (\ref{eq:horizontal_payoff_corner_n_workers}) this gives
\begin{equation*}
    \begin{split}
       \tilde{U_V} - U_H^C  &= \frac{1}{\rho} -\frac{1}{\rho + H} \\
        &= \frac{H}{\rho(\rho+H)}>0
    \end{split}  
\end{equation*}
Thus the vertical structure's payoff dominates for all \(\kappa > \bar{m}\) for the corner solution.\\

\noindent \emph{Interior meeting time solutions.}  Now suppose the conditions for the corner solution do not hold, that is all horizontal meetings occur. Equality of the marginal precisions \(h^\prime\) gives
\begin{equation}
    \frac{\xi_1}{m_1+\rho\xi_1^2}=\frac{\xi}{m_j+\rho\xi^2}\quad j\geq 2.
\end{equation}
Hence there must exit a common constant \(C>0\) such that
\begin{equation}
    m_i = C\xi_i - \rho \xi_i^2.
\end{equation}
Summing over all the meetings gives
\begin{equation*}
    \begin{split}
        \sum_{i=1}^n m_i &= \sum_{i=1}^n (C\xi_i -\rho\xi_i^2) \\
        \bar{m} &= C(\xi_1+(n-1)\xi) - \rho(\xi_1^2 + (n-1)\xi^2).
    \end{split}
\end{equation*}
Define
\begin{equation*}
    S_n= \xi_1+(n-1)\xi, \qquad N_{H,n}= \bar{m}+ \rho(\xi_1^2 + (n-1)\xi^2),
\end{equation*}
so C becomes
\begin{equation*}
C=\frac{N_{H,n}}{S_n}.
\end{equation*}
This means the meetings of the workers can be written as
\begin{equation*}
    \begin{split}
        m_1^H &= \frac{\xi_1[\bar{m}+\rho(n-1)\xi d]}{S_n}\\
        m_j^H &= \frac{\xi[\bar{m}-\rho\xi_1 d]}{S_n}, \ j\geq 2.
    \end{split}
\end{equation*}
The second expression is positive under the condition \(\bar{m}>\rho\xi_1 d\). At the optimum the precision weights are given by 
\begin{equation*}
    h(m_i;\xi_i)= \frac{\rho(C\xi_i-\rho\xi_i^2)}{C\xi_i}= \rho -\frac{\rho^2\xi_i}{C}.
\end{equation*}
Summing over all horizontal the precisions gives and using \(C\)
\begin{equation*}
\begin{split}
   H_T= \sum_{i=1}^n\rho -\frac{\rho^2\xi_i}{C} &= n\rho - \frac{\rho^2(\xi_1 + (n-1)\xi)}{C}  \\
     &= n\rho- \frac{\rho^2S_n}{C}\\
     & = \frac{\rho(nN_{H,n}- \rho S_n^2) }{N_{H,n}}
\end{split}
\end{equation*}
Then \(S_n^2=\xi_1^2+2(n-1)\xi\xi_1+(n-1)^2\xi^2\), it can be shown that
\begin{equation*}
\begin{split}
    n[\xi_1^2+(n-1)\xi^2] - S_n^2 &=  n\xi_1^2+ n(n-1)\xi^2 - [\xi_1^2+2(n-1)\xi\xi_1+(n-1)^2\xi^2] \\
    & = (n-1)(\xi_1^2-2\xi\xi_1+\xi^2) \\
    &= (n-1)d^2.
\end{split}
\end{equation*}
Let \(D_{H,n}=n\bar{m} + \rho(n-1)d^2\), so that 
\begin{equation*}
\begin{split}
    H_T & = \frac{\rho(nN_{H,n} - \rho S_n^2) }{N_{H,n}} \\
    & = \frac{\rho(n[\bar{m}+ \rho(\xi_1^2 + (n-1)\xi^2)] - \rho S_n^2) }{N_{H,n}} \\
    & = \frac{\rho(n\bar{m} +\rho[n(\xi_1^2 + (n-1)\xi^2) - S_n^2] ) }{N_{H,n}} \\
    & = \frac{\rho(n\bar{m} +\rho[(n-1)d^2] ) }{N_{H,n}} \\
    & = \frac{\rho D_{H,n} }{N_{H,n}}.
\end{split}
\end{equation*}
Using \(nN_{H,n} - \rho S_n^2=D_{H,n}\)
\begin{equation}
\label{eq:horizontal_payoff_n_workers_optimal}
    U_H= \bar{u}- \frac{N_{H,n}}{\rho D_{H,n}} = \bar{u}- \frac{1}{n\rho} - \frac{ S_n^2}{nD_{H,n}}
\end{equation}
For the vertical structure, let \(m_{-1}=\kappa - m_1\) be the meeting time available to worker 1 to meet all workers \(2,\dots, n\). Using \(h(m_1:\xi)\), since workers are symmetric and worker 1 essentially plays the role of the DM in the horizontal structure, and \(h(m_1:\xi)\) is concave, worker 1 divides \(m_{-1}\) equally
\begin{equation*}
    m_j=\frac{m_{-1}}{n-1}, \quad j\geq2.
\end{equation*}
Define
\begin{equation*}
    \psi_n = \xi^2\rho(n-1).
\end{equation*}
Thus the total precision contributed by workers \(2,\dots,n\) is
\begin{equation*}
\begin{split}
    H_T &= \sum_{j=2}^n h(\frac{m_{-1}}{n-1}:\xi) \\
        &= \sum_{j=2}^n \frac{s\rho }{s + \psi_n} \\
        &= \frac{s\rho(n-1) }{s + \psi_n}.
\end{split}
\end{equation*}
Including worker 1's precision weight before the final meeting with the DM is
\begin{equation*}
    G_V(s) = \rho + H_T = \frac{\rho(ns+\psi_n)}{s + \psi_n}.
\end{equation*}
This gives vertical payoff function of
\begin{equation}
\label{eq:vertical_payoff_n_diff_comms}
    U_V(m_1) = \bar{u} - \frac{\kappa -m_1 + \psi_n }{\rho[n(\kappa-m_1)+\psi_n]} - \frac{\xi_1^2}{m_1}, \qquad 0<m_1\leq \bar{m}.
\end{equation}
Differentiating and substituting in \(\frac{\psi_n}{\rho} = \xi^2(n-1)\) yields
\begin{equation}
    \begin{split}
        U_V^\prime(m_1) &= - \frac{\xi^2(n-1)^2}{[n(\kappa-m_1)+\psi_n]^2} +\frac{\xi_1^2}{m_1^2} \\
        U_V^{\prime\prime}(m_1) & = - \frac{2\xi^2n(n-1)^2}{[n(\kappa-m_1)+\psi_n]^3} -\frac{2\xi_1^2}{m_1^3} <0.
    \end{split}
\end{equation}
Thus the vertical payoff is strictly concave. The unconstrained first-order condition is given by
\begin{equation*}
     \frac{\xi(n-1)}{[n(\kappa-m_1)+\psi_n]} =\frac{\xi_1}{m_1}, 
\end{equation*}
this can be rearranged to give the final meeting length
\begin{equation}
\label{eq:vertical_foc_meeting_1}
     m_1(\kappa) =\frac{\xi_1(n\kappa +\psi_n)}{n\xi_1+(n-1)\xi}. 
\end{equation}
The constrained optimum is
\begin{equation}
    m_1^*(\kappa) = \min\{\bar{m}, m_1(\kappa)\}.
\end{equation}
Let \(B_n=n\xi_1+(n-1)\xi\), the DM meeting constraint binds when \(m_1(\kappa)\geq \bar{m}\),  which can be written as
\begin{equation}
   \kappa \geq \kappa_{B,n} \equiv \frac{\bar{m}B_n}{\xi_1n} -\frac{\psi_n}{n}.
\end{equation}
this establishes the exact closed-form threshold $\kappa_{B,n}$ governing the transition between the interior and capacity-constrained vertical regimes. \\

\noindent Under the horizontal corner condition \(\bar{m} \leq \rho\xi_1 d \Rightarrow \bar{m} < \rho\xi_1 \xi  \), from equation (\ref{eq:vertical_foc_meeting_1}) this yields \(m_1(\bar{m}) \geq \bar{m}\) if and only if \(\bar{m} < \rho\xi_1 \xi \). So the unconstrained meeting exceeds the limit \(\kappa\) and vertical optimum \(m_1^*(\kappa=\bar{m})=\bar{m}\). The remaining meeting for workers \(2,\dots,n\) is therefore \(s=\kappa-m_1^*=0\): this implies no information aggregation occurs. Substituting \(m_1^*=\kappa=\bar{m}\) into the Vertical payoff gives
\begin{equation*}
    U_V(M) = \bar{u}- \frac{\psi_n}{\rho\psi_n} - \frac{\xi_1^2}{\bar{m}} = \bar{u}- \frac{1}{\rho} - \frac{\xi_1^2}{\bar{m}} = U_H^C.
\end{equation*}
So the structures have equal payoffs. This establishes (a) \(\bar{\kappa}_n=\bar{m}\) \\

\noindent Now suppose the worker 1 meeting with the DM is not at the constraint, \(\kappa<\kappa_{B,n} \), so that \(m_1^*=m_1(\kappa) < \bar{m} \). Let \(D_V=n\kappa+\psi_n\), then at \(m_1=m_1(\kappa)\) using \(B_n\) and equation \(\ref{eq:vertical_foc_meeting_1}\),
\begin{equation*}
    \begin{split}
        n(\kappa-m_1)+\psi_n &= n\kappa+\psi_n - nm_1(\kappa) \\
        & = D_V - \frac{n\xi_1 D_V}{B_n} \\
        & = D_V \frac{(n-1)\xi}{B_n}.
    \end{split}
\end{equation*}
Moreover,
\begin{equation*} 
\begin{split} \frac{\kappa-m_1+\psi_n}{\rho[n(\kappa-m_1)+\psi_n]}-\frac1{n\rho} &=\frac{n(\kappa-m_1+\psi_n)-[n(\kappa-m_1)+\psi_n]} {n\rho[n(\kappa-m_1)+\psi_n]}\\ &=\frac{(n-1)\psi_n}{n\rho[n(\kappa-m_1)+\psi_n]}\\ &=\frac{\xi^2(n-1)^2}{n[n(\kappa-m_1)+\psi_n]}, 
\end{split} 
\end{equation*}
giving
\begin{equation*}
    \frac{\kappa-m+\psi_n}{\rho[n(\kappa-m)+\psi_n]}=\frac1{n\rho}+\frac{\xi^2(n-1)^2}{n[n(\kappa-m)+\psi_n]}.
\end{equation*}
Multiply the first-order condition with \(\xi_1\) to get
\begin{equation*}
         \frac{\xi\xi_1(n-1)}{[n(\kappa-m_1)+\psi_n]} =\frac{\xi_1^2}{m_1}.
\end{equation*}
Combining the terms gives a vertical payoff (\ref{eq:vertical_payoff_n_diff_comms}) gives
\begin{equation*} 
\begin{split} 
U_V^I(m_1) &= \bar{u} - \frac{\kappa -m_1 + \psi_n }{\rho[n(\kappa-m_1)+\psi_n]} - \frac{\xi_1^2}{m_1} \\
        &= \bar{u} - \frac1{n\rho}-\frac{\xi^2(n-1)^2}{n[n(\kappa-m)+\psi_n]} -  \frac{\xi\xi_1(n-1)}{[n(\kappa-m_1)+\psi_n]} \\
        & = \bar{u} - \frac1{n\rho} \frac{\xi^2(n-1)^2+ \xi\xi_1 n(n-1) }{n[ \frac{D_V(n-1)\xi}{B_n}]} \\
        & = \bar{u} - \frac1{n\rho} - B_n \frac{\xi(n-1)+ \xi_1 n }{n[D_V]} \\
        & = \bar{u} - \frac1{n\rho} - \frac{B_n^2 }{n[n\kappa + \psi_n]}
\end{split} 
\end{equation*}
Equating this with the horizontal payoff from equation (\ref{eq:horizontal_payoff_n_workers_optimal}) yields
\begin{equation*}
    \frac{B_n^2}{n\kappa + \psi_n} = \frac{S_n^2}{D_{H,n}}.
\end{equation*}
Thus the \(\kappa<\kappa_{B,n}\) at which the vertical structure has a higher payoff for an interior worker 1 and DM meeting i.e. \(m_1^*<\bar{m}\) is given by
\begin{equation}
  \kappa_{I,n} = \frac{D_{H,n}\big(\frac{B_n}{S_n}\big)^2-\psi_n}{n}.
\end{equation}
Hence, \(\bar{\kappa}_n= \kappa_{I,n}\). \\

\noindent Finally, suppose that \(\kappa\geq\kappa_{B,n}\), so now \(m_1^*=\bar{m}\). Let \(t=\kappa-\bar{m}\). The vertical payoff is now
\begin{equation}
\label{eq:vertical_payoff_n_workers_optimal_C}
    U_V^C(t) = \bar{u} - \frac{t+\psi_n}{\rho(n t +\psi_n)} - \frac{\xi_1^2}{\bar{m}}
\end{equation}
Define
\begin{equation*}
    F_n = \bar{m} \bigl[(n+1)\xi_1+(n-1)\xi\bigr]-n\rho\xi_1^2d
\end{equation*}
Under \(\bar{m} > \rho \xi_1 d\), 
\begin{equation*}
\begin{split}
        F_n &> \rho \xi_1 d\bigl[(n+1)\xi_1+(n-1)\xi\bigr]-n\rho\xi_1^2d \\
        & =\rho \xi_1 d\bigl[\xi_1+(n-1)\xi\bigr]
\end{split}
\end{equation*}
Subtract the horizontal payoff (\ref{eq:horizontal_payoff_n_workers_optimal}) from the capped vertical payoff (\ref{eq:vertical_payoff_n_workers_optimal_C}) 
\begin{equation*}
   U_V^C(t)-U_H = \frac{1}{n\rho} + \frac{S_n^2}{nD_{H,n}}
-\frac{t+\psi_n}{\rho(nt+\psi_n)}-\frac{\xi_1^2}{\bar{m}}.
\end{equation*}
Then using \(\psi_n=\rho(n-1)\xi^2\)
\begin{equation*}
    \begin{split}
        \frac{t+\psi_n}{\rho(nt+\psi_n)} &= \frac{1}{n\rho} - \frac{1}{n\rho} + \frac{t+\psi_n}{\rho(nt+\psi_n)} \\
        &= \frac{1}{n\rho}  + \frac{n(t+\psi_n)-(nt+\psi_n)}{n\rho(nt+\psi_n)} \\
        &= \frac{1}{n\rho}  + \frac{(n-1)\psi_n}{n\rho(nt+\psi_n)} \\
        &= \frac{1}{n\rho}  + \frac{\xi^2(n-1)^2}{n(nt+\psi_n)}. \\
    \end{split}
\end{equation*}
Thus
\begin{equation*}
 \begin{split}
   U_V^C(t)-U_H &= \frac{S_n^2}{nD_{H,n}}
- \frac{\xi^2(n-1)^2}{n(nt+\psi_n)}-\frac{\xi_1^2}{\bar{m}} \\
& = \frac{\bar{m}(nt+\psi_n)S_n^2 -\bar{m}D_{H,n}\xi^2(n-1)^2-n\xi_1^2(nt+\psi_n)D_{H,n}}{n\bar{m}(nt+\psi_n)D_{H,n}} \\
&=\frac{nt\big(\bar{m} S_n^2-n\xi_1^2D_{H,n}\bigr) +
\bigl(\bar{m} \psi_nS_n^2 -\bar{m}D_{H,n}\xi^2(n-1)^2 -n\xi_1^2\psi_nD_{H,n}\big)}{n\bar{m}(nt+\psi_n)D_{H,n}}
 \end{split}
\end{equation*}
Replace coefficient on \(t\) with
\begin{equation*}
\begin{split}
\bar{m} S_n^2-n\xi_1^2D_{H,n} =&\bar{m}\left[\xi_1^2+2(n-1)\xi_1\xi +(n-1)^2\xi^2\right]\\
&-n\xi_1^2\left[n\bar{m}+\rho(n-1)d^2\right]\\
=&(n-1)d\left(
\bar{m}[(n+1)\xi_1+(n-1)\xi] -n\rho\xi_1^2d \right)\\
=&(n-1)dF_n.
\end{split}
\end{equation*}
The substitute in \(\psi_n=\rho(n-1)\xi^2\) and \(D_{H,n}=n\bar{m}+\rho(n-1)d^2\) to get 
\begin{equation*}
    U_V^C(t)-U_H = \frac{(n-1)\left[F_n d t-\xi^2(n-1)(\bar{m}-\rho\xi_1d)^2\right]}{\bar{m}(nt+\psi_n)D_{H,n}}.
\end{equation*}
The denominator is positive and so vertical structure has a higher payoff when
\begin{equation*}
    t > \frac{\xi^2(n-1)(\bar{m}-\rho\xi_1d)^2}{F_n d}.
\end{equation*}
Given \(t=\kappa-\bar{m}\), the switch occurs at
\begin{equation*}
    \kappa_{C,n}= \bar{m} + \frac{\xi^2(n-1)(\bar{m}-\rho\xi_1d)^2}{F_n d}.
\end{equation*}
Note that the vertical payoff formula changes at \(\kappa_{B,n}\), when the meeting between worker 1 and DM meets the DM's constraint. Since the crossing \(\kappa_{I,n}\) is calculated using the interior formula, it can be that \(\kappa_{I,n}> \kappa_{B,n}\), it is a formal root of the interior expression, but it does not describe the actual crossing: that expression no longer applies and instead \(\kappa_{C,n}\) describes the actual crossing threshold.\\

\noindent To check that \(\kappa_{C,n}\) is a unique switch point over the feasible \(\kappa\geq \bar{m}\). Need to show that the vertical structure's payoff is below horizontal at the smallest feasible \(\kappa\), eventually rises above it, and crosses it only once.

\medskip
\noindent\emph{Existence and uniqueness of the crossing.}
Compare the structures at the smallest feasible value,
\(\kappa=\bar{m}\). In the horizontal-corner region, \(\bar{m}\leq \rho \xi_1  d\), the vertical cap binds at \(\kappa=\bar{m}\), and the two structures have the same payoff. For every \(\kappa>\bar{m}\), the feasible allocation
constructed after equation
\eqref{eq:horizontal_payoff_corner_n_workers} gives a strictly higher
vertical payoff. Hence the threshold in this region is
\(\bar{\kappa}_n=\bar{m}\).

Now consider the interior-horizontal region, \(\bar{m}>\rho \xi_1  d\). If the
vertical cap binds at \(\kappa=\bar{m}\), the vertical payoff is the
worker-1-only payoff
\begin{equation*}
U_V^*(\bar{m})=\bar{u}-\frac{1}{\rho}-\frac{\xi_1^2}{\bar{m}}.
\end{equation*}
Direct comparison with equation
\eqref{eq:horizontal_payoff_n_workers_optimal} gives
\(U_H>U_V^*(\bar{m})\). If instead the vertical solution is interior at
\(\kappa=\bar{m}\), then \(\bar{m}>\rho \xi_1 \xi\). Let \(q=\bar{m}-\rho \xi_1 \xi>0\). Subtracting
the horizontal payoff from the interior vertical loss, or equivalently
subtracting the interior vertical payoff from the horizontal payoff,
gives
\begin{equation*}
\begin{aligned}
U_H-U_V^I(\bar{m})
&=L_V^I(\bar{m})-L_H\\
&=
\frac{\xi_1 (n-1)}{(n\bar{m}+\psi_n)D_{H,n}}
\Bigl(
q[(n+1)\xi_1 +2(n-1)\xi]
+\rho \xi_1^2[n\xi_1 +(n-1)\xi]
\Bigr)>0.
\end{aligned}
\end{equation*}
Thus, throughout the interior-horizontal region, horizontal
communication has the higher payoff at \(\kappa=\bar{m}\).

Next, consider large \(\kappa\). On the capped branch,
\(t=\kappa-\bar{m}\) tends to infinity, so equation
\eqref{eq:vertical_payoff_n_workers_optimal_C} implies
\begin{equation*}
\lim_{\kappa\to\infty}U_V^*(\kappa)
=
\bar{u}-\frac{1}{n\rho}-\frac{\xi_1^2}{\bar{m}}.
\end{equation*}
Using equation \eqref{eq:horizontal_payoff_n_workers_optimal},
\begin{equation*}
U_V^*(\infty)-U_H = \frac{(n-1)dF_n}{n\bar{m}D_{H,n}}>0,
\end{equation*}
where \(F_n>0\) was established above. Therefore, vertical
communication eventually has the higher payoff.

It remains to show that the optimized vertical payoff increases
strictly with \(\kappa\). For a fixed feasible meeting length
\(m_1\), write \(s=\kappa-m_1\). The vertical loss component that
depends on \(s\) satisfies
\begin{equation*}
\frac{\partial}{\partial s}
\left(\frac{s+\psi_n}{\rho(ns+\psi_n)}\right) =
-\frac{(n-1)\psi_n}{\rho(ns+\psi_n)^2}<0. \end{equation*}
Thus, holding \(m_1\) fixed, increasing \(\kappa\) strictly lowers
vertical loss and strictly raises vertical payoff. The optimized
vertical payoff is therefore strictly increasing in \(\kappa\).
It is continuous when the optimum reaches the cap:
the interior and capped formulas coincide at
\(\kappa=\kappa_{B,n}\), where \(m_1^*(\kappa)=\bar{m}\).

Consequently, in the interior-horizontal region, the continuous,
strictly increasing optimized vertical payoff starts below \(U_H\)
at \(\kappa=\bar{m}\) and eventually exceeds \(U_H\). It crosses \(U_H\)
exactly once. If the interior root \(\kappa_{I,n}\) is no greater
than \(\kappa_{B,n}\), the crossing is \(\kappa_{I,n}\); otherwise
the crossing lies on the capped branch and is \(\kappa_{C,n}\).
Combining this with the horizontal-corner case gives
\begin{equation*}
\bar{\kappa}_n=
\begin{cases}
\bar{m}, & \bar{m}\leq \rho \xi_1 d,\\
\kappa_{I,n}, & \bar{m}>\rho \xi_1 d \text{ and } \kappa_{I,n}\leq\kappa_{B,n},\\
\kappa_{C,n}, & \bar{m}>\rho \xi_1  d \text{ and } \kappa_{I,n}>\kappa_{B,n}.
\end{cases}
\end{equation*}
Therefore, for feasible \(\kappa\geq \bar{m}\), vertical communication
is strictly preferred if and only if
\(\kappa>\bar{\kappa}_n\).

\end{proof}

\printbibliography

\newpage

\end{document}